\PassOptionsToPackage{dvipsnames}{xcolor}
\documentclass{article}

\usepackage{amsmath,amsthm,amsfonts,amssymb}
\usepackage{todonotes}
\usepackage{mathtools}
\usepackage{latexsym}
\usepackage{authblk}
\usepackage{tikz,tikz-cd}
\usepackage{comment}
\usepackage{enumitem}
	\setlist[enumerate]{label=(\roman*)}  
	\setlist[enumerate,2]{label=(\alph*)}  
\usepackage[T1,T2A]{fontenc}			

\usepackage{tikzit}				
\definecolor{pice}{HTML}{c0d3e1} 
\colorlet{boxcolor}{pice!64}
\tikzstyle{bn}=[fill=black, draw=black, shape=circle, inner sep=1.5pt]
\tikzstyle{medium box}=[fill=boxcolor, draw=black, shape=rectangle, minimum width=0.7cm, minimum height=0.7cm]
\tikzstyle{horiz state}=[fill=boxcolor, draw=black, regular polygon, regular polygon sides=3, minimum width=1cm, shape border rotate=90, inner sep=0pt]
\tikzstyle{arrow}=[->]
\tikzstyle{dashed box}=[-, dashed]
\tikzstyle{over arrow}=[-, black, preaction={draw=white, double}]

\definecolor{myurlcolor}{rgb}{0,0,0}
\definecolor{mycitecolor}{rgb}{0,0,0}
\definecolor{myrefcolor}{rgb}{0,0,0}
\usepackage[pagebackref,draft=false]{hyperref}
\hypersetup{colorlinks,
linkcolor=myrefcolor,
citecolor=mycitecolor,
urlcolor=myurlcolor}

\usepackage[capitalize]{cleveref}

\allowdisplaybreaks 

\newtheorem{theorem}{Theorem}[section]
\newtheorem{proposition}[theorem]{Proposition}
\newtheorem{lemma}[theorem]{Lemma}
\newtheorem{corollary}[theorem]{Corollary}
\newtheorem{definition}[theorem]{Definition}
\newtheorem{question}[theorem]{Question}

\theoremstyle{definition}

\newtheorem{remark}[theorem]{Remark}

\newtheorem{convention}[theorem]{Convention}

\newcommand{\cat}[1]{{\mathsf{#1}}} 

\newcommand{\id}{\mathrm{id}} 		

\newcommand{\R}{\mathbb{R}}		

\tikzset{pullback/.style={minimum size=1.2ex,path picture={	
			\draw[opacity=1,black,-,#1] (-0.5ex,-0.5ex) -- (0.5ex,-0.5ex) -- (0.5ex,0.5ex);%
}}}

\newcommand{\et}{\parallel}

\DeclareMathOperator{\cop}{copy}
\DeclareMathOperator{\discard}{del}

\title{Markov Categories and Entropy}

\author{Paolo Perrone}

\affil{University of Oxford,\\Department of Computer Science}

\date{}

\begin{document}

\maketitle

\begin{abstract}
 Markov categories are a novel framework to describe and treat problems in probability and information theory.
 In this work we combine the categorical formalism with the traditional quantitative notions of entropy, mutual information, and data processing inequalities. 
 We show that several quantitative aspects of information theory can be captured by an enriched version of Markov categories, where the spaces of morphisms are equipped with a divergence or even a metric.
 
 Following standard practices of information theory, we get measures of mutual information by quantifying, with a chosen divergence, how far a joint source is from displaying independence of its components. 
 
 More strikingly, Markov categories give a notion of determinism for sources and channels, and we can define entropy exactly by quantifying how far a source or channel is from being deterministic. This recovers Shannon and Rényi entropies, as well as the Gini-Simpson index used in ecology to quantify diversity, and it can be used to give a conceptual definition of generalized entropy.
 
 No previous knowledge of category theory is assumed.
\end{abstract}

\tableofcontents

\section*{Introduction}
\addcontentsline{toc}{section}{Introduction}

In this work we integrate two main themes of information theory. On one hand there is a qualitative description of information flow, for example by means of graphical representation of the stochastic dependence relations, or by means of category-theoretic ideas.
On the other hand there is quantitative reasoning, based on measures such as entropy and mutual information, and on inequalities such as data processing inequalities.
We can incorporate the quantitative aspects into the categorical framework using the theory of \emph{enriched} categories. (Its previous knowledge is however not required to understand this work.)

Since the early days of information theory there has been interest in categorical structures to describe probabilistic processes (the first published reference seems to be due to Čencov~\cite{chentsov}).
Recently, there has been growing interest in \emph{Markov categories}, defined in their current form by Fritz in \cite{fritz2019synthetic}.\footnote{Markov categories are related to older structures called ``copy-discard'' or ``garbage-share'' categories~\cite{gadducci-thesis,chojacobs2019strings}. See \cite[Remark~2.2]{fritz2022free} for a detailed history of the concept.} 
They can be seen as an abstraction of categories of kernels, which come equipped with a graphical calculus representing the information flow faithfully.
Indeed, the graphical calculus of Markov categories is known to satisfy a \emph{d-separation theorem}~\cite{fritz2022dseparation}, and hence can be thought of as a general theory of probabilistic graphical models, alongside Bayesian networks and Markov random fields.
There is a correspondence between the graphical representation and the mathematical structures that allow us to prove theorems simply by graphical manipulations. 

Several theorems of probability theory and related fields have been reproven in this way, and sometimes generalized. Among these results, several theorems on sufficient statistics \cite{fritz2019synthetic,jacobs-statistics}, the zero-one laws of Kolmogorov and Hewitt-Savage \cite{fritzrischel2019zeroone}, the Blackwell-Sherman-Stein theorem on comparison of statistical experiments \cite{fritz2020representable}, de Finetti's theorem~\cite{fritz2021definetti,ours_LICS}, and the ergodic decomposition theorem~\cite{ergodic}.
Markov categories have also been used to model aspects of information flow \cite{fritz2022dilations}, capturing several qualitative concepts of information theory, such as dependence and independence, and signalling. 

In this work we turn to more quantitative concepts of information theory, in particular divergences and entropy, and show how they fit into the formalism of Markov categories. 
In order to incorporate quantitative statements into the categorical formalism we make use of \emph{enriched category theory}~\cite{basicconcepts}, a version of category theory where the set of arrows between any two objects is replaced by a more general structure. In our case, we take a metric or divergence space, where we can measure ``how far'' two morphisms are from being equal, or equivalently, ``how far'' a diagram is from commuting. 
While at first it might seem that metrics have more desirable properties than more general divergences, our formalism will work in general.
We focus on three choices of divergences: the \emph{Kullback-Leibler divergence} (or \emph{relative entropy}), the more general \emph{Rényi divergences}, and the \emph{total variation distance}.

It is customary, in information theory, to define mutual information as a measure of departure from the case of stochastic independence. 
This fits very well into the Markov categories formalism, where there is a native, abstract notion of stochastic independence, based on equality of two suitably constructed morphisms (see \Cref{mi}). By measuring the departure from this case, one can reconstruct exactly measures such as Shannon and Rényi's mutual information.

Markov categories also come with a notion of \emph{determinism}, again based on an equation between morphisms (see \Cref{H}). By measuring the departure from this case, and choosing our divergences appropriately, we can recover exactly Shannon and Rényi's entropies, and from the total variation distance one obtains the \emph{Gini-Simpson index}, used for example in ecology to quantify diversity~\cite{leinster-entropy}. Our approach therefore gives an equivalent, abstract definition of (generalized) entropy, at least for the discrete case.

\paragraph{Previous work on category theory and entropy.}
Entropy and its properties have often been of interest for the category theory community. In \cite{entropy-loss} \cite{short-entropy} and \cite{infloss-stoch}, Shannon's entropy was given a categorical characterization formalizing the idea of measuring information loss.
In \cite{entropy-bayes}, relative entropy (the KL divergence) was given a characterization in terms of Bayesian inference for the discrete case, and in \cite{entropy-sbs} for the general Standard Borel case. A 2-dimensional generalization was given in \cite{2-relative-entropy}.
Entropy can be studied through the lens of the operad of convex spaces, and in \cite{entropy-operads} it was shown to be a derivation on such an operad. Its nature as a derivation has also been explored from the point of view of homology in \cite{homology-entropy}.
The compositional properties of entropy have also been studied in terms of polynomial functors~\cite{polynomial-entropy}.
On the quantum side, von Neumann entropy was given a categorical characterization in~\cite{entropy-vn}, and there is work on a characterization of quantum relative entropy~\cite{entropy-quantum}.
From the point of thermostatics and thermodynamics, there is work on the categorical significance of entropy and related quantities~\cite{compositional-thermostatics}.
From an algebraic perspective, a categorical generalization of the concept of algebraic entropy has been given in \cite{entropy-in-a-cat}.
Also, both classical and quantum entropy, and their relation to contextuality, have been explored in \cite{topos-entropy}.

This work is not the first approach to entropy which combines metric geometry and category theory. The first ideas on the matter seem to be due to Gromov~\cite{gromov-entropy} and have inspired, besides this work, a number of other independent approaches, such as \emph{tropical probability theory}~\cite{tropical-probability1,tropical-probability2,tropical-probability3,tropical-probability4,tropical-probability5}.

Finally, category theory, metric geometry, and entropy are also main themes of the book \cite{leinster-entropy}, about entropy-like quantities used as measures of diversity, for example in the context of ecology.

\paragraph{Outline of this work.}
In \Cref{background} we give an overview of Markov categories, focusing on the two main examples used in this work, the category $\cat{FinStoch}$ of finite alphabets and stochastic matrices (noisy channels) between them, and the category $\cat{Stoch}$ of infinite measurable alphabets and Markov kernels between them.

In \Cref{divmcat} we review the notion of divergence, or statistical distance, and we define an enrichment on Markov categories (\Cref{defdivmarkov}). We give an interpretation of the inequalities involved, in particular, a data processing inequality (\Cref{qconvdproc}). By reviewing the Markov-categorical notion of joints and marginals, we give an equivalent characterization of enrichment in terms of them, which can be seen as a monotonicity condition in the number of observed variables, together with a generalized chain rule (\Cref{divjm}).
We then turn to particular examples, where we show that the KL divergence (relative entropy), the Rényi $\alpha$-divergences, and the total variation distance all give enrichments on the Markov categories $\cat{Stoch}$ and $\cat{FinStoch}$. We also show that in general, the Tsallis $q$-divergences do not give an enrichment. 
In \Cref{partitions} we show that in our examples, the divergence between nondiscrete probability measures can be expressed as a supremum over countable partitions, and express the result as an enriched universal property, the first one in our formalism. 
In \Cref{conddiv} we then define a conditional version of divergences, which can be seen as a measure of departure from almost-sure equality of channels.

In \Cref{mi} we review the notion of independence and conditional independence in Markov categories, and define mutual information as a measure of departure from the independence case. This is in line with the traditional information-theoretic approach, and it recovers the usual notions of Shannon mutual information and $\alpha$-mutual information for their corresponding divergences (\Cref{micases}).
We show that all these measures of mutual information, by construction, satisfy a data processing inequality (\Cref{datapmi}), which once again implies a monotonicity condition in the number of observed variables.
We also show that our measure of conditional divergence, quantifying the departure from almost sure conditional independence, recovers classical measures of conditional mutual information (\Cref{condmi}).

In \Cref{H} we review the notion of deterministic sources and channels in a Markov category, and define entropy as a measure of departure from determinism. 
This recovers some well known measures of randomness in the discrete case (\Cref{Hcasefin}). In particular, the KL divergence gives Shannon entropy, the Rényi $\alpha$-divergence gives the Rényi entropy, but of a different order ($2-\alpha$), and the total variation distance gives the Gini-Simpson index. 
Similarly, measuring the departure from almost sure determinism gives us conditional entropy (\Cref{condent}).
In the nondiscrete case, these measures of entropy are all maximal for atomless distributions (\Cref{Hcaseinf}). We argue that this is due to the fact that measurable spaces are insufficient to describe sources and channels in the continuous case, and suggest a more geometrical approach (\Cref{futurework}).

Finally, in \Cref{Div} we spell out the details of the category of divergence spaces, which we are using as enrichment to our Markov categories.
The content of the appendix, or any previous knowledge of enriched category theory, is not required to understand the rest of this work.

\paragraph{Acknowledgements.}
The author would like to thank Tobias Fritz, Tomá\v{s} Gonda and Sam Staton for the insightful discussions.

\section{Background: Markov categories}\label{background}

A Markov category is an abstraction of a system of noisy processing units and data that they can share as input and output.

\paragraph{Alphabets and channels.}
First of all, a Markov category consists of a collection of \emph{objects}, denoted by $X$, $Y$, and so on, which we think of as spaces of possible states or data, or alphabets. We represent them as wires, in this work, horizontal.
\[
 \begin{tikzpicture}
	\begin{pgfonlayer}{nodelayer}
		\node [style=none] (0) at (-1, 0) {};
		\node [style=none] (1) at (1, 0) {};
		\node [style=none] (2) at (0, 0.5) {$X$};
	\end{pgfonlayer}
	\begin{pgfonlayer}{edgelayer}
		\draw (0.center) to (1.center);
	\end{pgfonlayer}
\end{tikzpicture}

\]
In this work we will mostly consider as objects either \emph{finite alphabets}, which will form the Markov category $\cat{FinStoch}$, or possibly infinite, \emph{measurable alphabets}, which will form the Markov category $\cat{Stoch}$. 
The objects of $\cat{FinStoch}$ are finite sets, and the objects of $\cat{Stoch}$ are measurable spaces. We denote a measurable space by $(X,\Sigma_X)$ (where $\Sigma_X$ is the $\sigma$-algebra), or more briefly by $X$ when it does not cause ambiguity. 

Between two objects $X$ and $Y$ we can have \emph{morphisms} $f:X\to Y$, which we can interpret as channels, devices, or programs, which are in general noisy, involving randomness. We represent them as boxes to be read horizontally from left to right.
\[
 \begin{tikzpicture}
	\begin{pgfonlayer}{nodelayer}
		\node [style=none] (7) at (-2, 0) {};
		\node [style=none] (10) at (2, 0) {};
		\node [style=none] (11) at (-2.5, 0) {$X$};
		\node [style=none] (12) at (2.5, 0) {$Y$};
		\node [style=medium box] (13) at (0, 0) {$f$};
	\end{pgfonlayer}
	\begin{pgfonlayer}{edgelayer}
		\draw (7.center) to (10.center);
	\end{pgfonlayer}
\end{tikzpicture}

\]
In $\cat{FinStoch}$, a channel is a stochastic matrix from $X$ to $Y$, i.e.~a matrix of nonnegative entries with columns indexed by the elements of $X$, and rows indexed by the elements of $Y$,
\[
 \begin{tikzcd}[row sep=0]
  X \times Y \ar{r}{f} & {[0,1]} \\
  (x,y) \ar[mapsto]{r} & f(y|x)
 \end{tikzcd}
\]
such that each column sums to one,
\[
 \sum_{y\in Y} f(y|x) = 1 \qquad\mbox{for every } x\in X .
\]
We can interpret $f(y|x)$ as a conditional or transition probability from state $x\in X$ to state $y\in Y$, or we can interpret $f$ as a family of probability measures $f_x$ over $Y$ indexed by the elements of $X$. That is, if we denote by $PY$ the set of probability measures on $Y$, a stochastic matrix $f$ can equivalently be seen as a function 
\begin{equation}\label{fx}
 \begin{tikzcd}[row sep=0]
  X \ar{r} & PY \\
  x \ar[mapsto]{r} & f_x .
 \end{tikzcd}
\end{equation}

In $\cat{Stoch}$, a morphism $f:X\to Y$ is a Markov kernel from $X$ to $Y$, by which we mean an assignment 
\[
 \begin{tikzcd}[row sep=0]
  X \times \Sigma_Y \ar{r}{f} & {[0,1]} \\
  (x,S) \ar[mapsto]{r} & f(S|x) ,
 \end{tikzcd}
\]
which is measurable in the first argument, and which is a probability measure in the second argument. 
Just as for stochastic matrices, we can also view a kernel equivalently as a function in the form \eqref{fx}, which assigns to each $x\in X$ a probability measure $f_x\in PY$. This function defines a kernel if and only if it is measurable in $x$ for a suitably defined $\sigma$-algebra on $PY$ (see \cite{giry} for more).
Given a measurable (deterministic) function $f:X\to Y$, we can always obtain a kernel $K_f$ from $X$ to $Y$ as follows: for each $x\in X$ and $S\in\Sigma_Y$,
\[
 K_f(S|x) \coloneqq \delta_{f(x)}(S) = 1_S(f(x)) = \begin{cases}
                                    1 & f(x)\in S \\
                                    0 & f(x)\notin S .
                                   \end{cases}
\]
These can be seen as channels with no noise. 

We can model probability measures as channels with no inputs, as follows.
First of all, we have a distinguished object called the \emph{unit}, which we write $I$, and which we do not draw (it's represented by an empty region). It represents a situation of no information. In $\cat{FinStoch}$ and $\cat{Stoch}$ it is the one-point space, where there is no distinction between states to be made.
A \emph{source}, or (random) \emph{state} on $X$ is now a morphism $p:I\to X$, which we depict as follows.
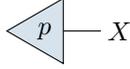
\[
 \begin{tikzpicture}
	\begin{pgfonlayer}{nodelayer}
		\node [style=none] (7) at (-1, 0) {};
		\node [style=none] (10) at (0.5, 0) {};
		\node [style=none] (11) at (1, 0) {$X$};
		\node [style=horiz state] (25) at (-1, 0) {$p$};
	\end{pgfonlayer}
	\begin{pgfonlayer}{edgelayer}
		\draw (7.center) to (10.center);
	\end{pgfonlayer}
\end{tikzpicture}

\]
In $\cat{FinStoch}$, a source is a stochastic matrix on $X$ of one column, i.e.~a finite probability measure on $X$.
In $\cat{Stoch}$ it is a Markov kernel to $X$ with no input, i.e.~a probability measure on the measurable space $X$.

\paragraph{Identities and sequential composition.}
The fact that we have a \emph{category} means the following.
First of all, we have an \emph{identity morphism} $\id_X:X\to X$ for each object (alphabet) $X$, which represents no change in the state of $X$. We draw it simply with a wire:
\[
 \begin{tikzpicture}
	\begin{pgfonlayer}{nodelayer}
		\node [style=none] (0) at (-1, 0) {};
		\node [style=none] (1) at (1, 0) {};
		\node [style=none] (2) at (-1.5, 0) {$X$};
		\node [style=none] (3) at (1.5, 0) {$X$};
	\end{pgfonlayer}
	\begin{pgfonlayer}{edgelayer}
		\draw (0.center) to (1.center);
	\end{pgfonlayer}
\end{tikzpicture}

\]
In $\cat{FinStoch}$, identities are identity matrices. In $\cat{Stoch}$ they are the ``Dirac delta'' kernels defined by the identity function,
\[
 \id(S|x) = \delta_x(S) = 1_S(x) = \begin{cases}
                                    1 & x\in S \\
                                    0 & x\notin S 
                                   \end{cases}
\]
for each $x\in X$ and $S\in\Sigma_X$.

Moreover, we have a notion of \emph{sequential composition} of channels: given channels $f:X\to Y$ and $g:Y\to Z$, we can form a channel $g\circ f:X\to Z$, which we draw as follows.

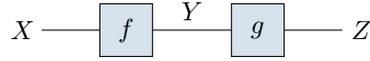
\[
 \begin{tikzpicture}
	\begin{pgfonlayer}{nodelayer}
		\node [style=none] (7) at (-4, 0) {};
		\node [style=none] (10) at (4, 0) {};
		\node [style=none] (11) at (-4.5, 0) {$X$};
		\node [style=none] (12) at (0, 0.5) {$Y$};
		\node [style=medium box] (13) at (1.75, 0) {$g$};
		\node [style=medium box] (14) at (-1.75, 0) {$f$};
		\node [style=none] (15) at (4.5, 0) {$Z$};
	\end{pgfonlayer}
	\begin{pgfonlayer}{edgelayer}
		\draw (7.center) to (10.center);
	\end{pgfonlayer}
\end{tikzpicture} 
\]
In $\cat{FinStoch}$ the composition is given by the \emph{Chapman-Kolmogorov formula}:
\[
 g\circ f\,(z|x) \coloneqq \sum_{y\in Y} g(z|y)\,f(y|x) ,
\]
and in $\cat{Stoch}$ it is given by its continuous analogue: for every measurable subset $S\subseteq Z$,
\[
 g\circ f\,(S|x) \coloneqq \int_Y g(S|y)\,f(dy|x) ,
\]
by which we mean the integral with respect to the measure $f_x$ on $Y$, for every $x$.
This makes the transitions $f$ and $g$ independent, as in a Markov process (hence the name, ``Markov category''), which models for example connecting devices whose sources of noise are independent. 
(Markov categories can also model more general, non-Markov stochastic processes, by means of \emph{joint sources and morphisms}, see \Cref{divjm} for more, as well as the original source \cite{fritz2019synthetic}.)

To have a \emph{category}, we have to require that this composition is associative, and that the identities behave indeed like identities. This is the case in $\cat{Stoch}$ and $\cat{FinStoch}$, as it is well known. 

\paragraph{Parallel composition.}
Markov categories also come with a notion of \emph{parallel} composition. 
First of all, given objects $X$ and $A$, we want a \emph{tensor product} object, which we denote by $X\otimes A$, and which we interpret as the object whose states are composite states. For example, in $\cat{FinStoch}$ and in $\cat{Stoch}$ it is given by the cartesian product of sets and of measurable spaces (the latter equipped with the product $\sigma$-algebra). 
Now given channels $f:X\to Y$ and $h:A\to B$, we can form the tensor product channel $f\otimes h:X\otimes A\to Y\otimes B$, which we represent as follows,
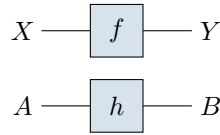
\[
 \begin{tikzpicture}
	\begin{pgfonlayer}{nodelayer}
		\node [style=none] (7) at (-2, 1) {};
		\node [style=none] (10) at (2, 1) {};
		\node [style=none] (11) at (-2.5, 1) {$X$};
		\node [style=none] (12) at (2.5, 1) {$Y$};
		\node [style=medium box] (13) at (0, 1) {$f$};
		\node [style=none] (14) at (-2, -1) {};
		\node [style=none] (15) at (2, -1) {};
		\node [style=none] (16) at (-2.5, -1) {$A$};
		\node [style=none] (17) at (2.5, -1) {$B$};
		\node [style=medium box] (18) at (0, -1) {$h$};
	\end{pgfonlayer}
	\begin{pgfonlayer}{edgelayer}
		\draw (7.center) to (10.center);
		\draw (14.center) to (15.center);
	\end{pgfonlayer}
\end{tikzpicture}
\]
and which we interpret as processing $X$ and $A$ independently.
Compare this with a generic channel $g:X\otimes A\to Y\otimes B$,
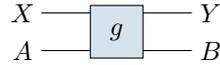
\[
 \begin{tikzpicture}
	\begin{pgfonlayer}{nodelayer}
		\node [style=none] (7) at (-2, 0.5) {};
		\node [style=none] (10) at (2, 0.5) {};
		\node [style=none] (11) at (-2.5, 0.5) {$X$};
		\node [style=none] (12) at (2.5, 0.5) {$Y$};
		\node [style=none] (14) at (-2, -0.5) {};
		\node [style=none] (15) at (2, -0.5) {};
		\node [style=none] (16) at (-2.5, -0.5) {$A$};
		\node [style=none] (17) at (2.5, -0.5) {$B$};
		\node [style=medium box] (18) at (0, 0) {$g$};
	\end{pgfonlayer}
	\begin{pgfonlayer}{edgelayer}
		\draw (7.center) to (10.center);
		\draw (14.center) to (15.center);
	\end{pgfonlayer}
\end{tikzpicture}
\] 
where for example, $Y$ can possibly depend on both $X$ and $A$.
In $\cat{FinStoch}$, the tensor product of the stochastic matrices $f:X\to Y$ and $g:A\to B$ is given by the product of the individual entries,
\[
 f\otimes h\,(y,b|x,a) \coloneqq f(y|x) \, h(b|a) ,
\]
and in $\cat{Stoch}$ it is defined analogously. 
In particular, for sources $p$ and $q$ on $X$ and $Y$, 
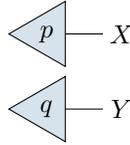
\[
 \begin{tikzpicture}
	\begin{pgfonlayer}{nodelayer}
		\node [style=none] (7) at (-1, 1) {};
		\node [style=none] (10) at (0.5, 1) {};
		\node [style=none] (11) at (1, 1) {$X$};
		\node [style=horiz state] (25) at (-1, 1) {$p$};
		\node [style=none] (26) at (-1, -1) {};
		\node [style=none] (27) at (0.5, -1) {};
		\node [style=none] (28) at (1, -1) {$Y$};
		\node [style=horiz state] (29) at (-1, -1) {$q$};
	\end{pgfonlayer}
	\begin{pgfonlayer}{edgelayer}
		\draw (7.center) to (10.center);
		\draw (26.center) to (27.center);
	\end{pgfonlayer}
\end{tikzpicture}
\]
the tensor product is just the product of the probabilities,
\[
 p\otimes q\,(x,y) = p(x)\,p(y) ,
\]
taken independently.

For technical reasons we require this tensor product to be associative and unital up to isomorphism (where the unit is given by the object $I$), and to be symmetric, i.e.~for all objects $X$ and $Y$ we need a distinguished isomorphism $X\otimes Y\cong Y\otimes X$, which we draw as follows.
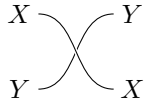
\[
 \begin{tikzpicture}
	\begin{pgfonlayer}{nodelayer}
		\node [style=none] (0) at (-1, 1) {};
		\node [style=none] (1) at (1, 1) {};
		\node [style=none] (2) at (-1, -1) {};
		\node [style=none] (3) at (1, -1) {};
		\node [style=none] (4) at (-1.5, 1) {$X$};
		\node [style=none] (5) at (-1.5, -1) {$Y$};
		\node [style=none] (6) at (1.5, 1) {$Y$};
		\node [style=none] (7) at (1.5, -1) {$X$};
	\end{pgfonlayer}
	\begin{pgfonlayer}{edgelayer}
		\draw [in=-180, out=0] (2.center) to (1.center);
		\draw [style=over arrow, in=-180, out=0] (0.center) to (3.center);
	\end{pgfonlayer}
\end{tikzpicture}
\]
In $\cat{Stoch}$ and $\cat{FinStoch}$, this morphism just switches the coordinates, $(x,y)\mapsto(y,x)$, with probability one.
These isomorphisms have to be compatible in such a way as to form what is called a \emph{symmetric monoidal category} (see for example \cite[Section~VII.1]{maclane-cwm1998} for more information).

\paragraph{Copy and discard.}
The last piece of structure that we need to form a Markov category is two distinguished maps for each object $X$: a map $\cop:X\to X\otimes X$ which we call ``copy'' or ``duplicate'', and represent as follows,
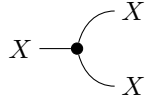
\[
 \begin{tikzpicture}
	\begin{pgfonlayer}{nodelayer}
		\node [style=none] (0) at (-1, 0) {};
		\node [style=none] (1) at (-1.5, 0) {$X$};
		\node [style=none] (2) at (1, 1) {};
		\node [style=none] (3) at (1, -1) {};
		\node [style=none] (4) at (1.5, 1) {$X$};
		\node [style=none] (5) at (1.5, -1) {$X$};
		\node [style=bn] (6) at (0, 0) {};
	\end{pgfonlayer}
	\begin{pgfonlayer}{edgelayer}
		\draw (0.center) to (6);
		\draw [in=180, out=75] (6) to (2.center);
		\draw [in=180, out=-75] (6) to (3.center);
	\end{pgfonlayer}
\end{tikzpicture}
\]
and a map $\discard:X\to I$ which we call ``delete'' or ``discard'', and represent as follows.
\[
 \begin{tikzpicture}
	\begin{pgfonlayer}{nodelayer}
		\node [style=none] (0) at (-1, 0) {};
		\node [style=none] (1) at (-1.5, 0) {$X$};
		\node [style=bn] (6) at (0, 0) {};
	\end{pgfonlayer}
	\begin{pgfonlayer}{edgelayer}
		\draw (0.center) to (6);
	\end{pgfonlayer}
\end{tikzpicture}

\]
As the names suggest, the two maps can be interpreted as copying and discarding the state of $X$. 
In both $\cat{Stoch}$ and $\cat{FinStoch}$, the copy map assigns to each $x\in X$ the point $(x,x)\in X\times X$ with probability one. In other words, it is the kernel defined by the diagonal embedding $X\to X\times X$. 
The discard map, in $\cat{Stoch}$ and $\cat{FinStoch}$, corresponds to summing (or integrating) the probabilities. For example, given a probability measure $p$ on $X$, we obtain the trivial probability measure $1$ on $I$ by summing,
\[
 \begin{tikzpicture}
	\begin{pgfonlayer}{nodelayer}
		\node [style=none] (0) at (-1, 0) {};
		\node [style=none] (1) at (-1.5, 0) {$X$};
		\node [style=bn] (6) at (0, 0) {};
	\end{pgfonlayer}
	\begin{pgfonlayer}{edgelayer}
		\draw (0.center) to (6);
	\end{pgfonlayer}
\end{tikzpicture}

 \qquad\qquad
 \sum_{x\in X} p(x) = 1 .
\]
Similarly, given a (joint) probability measure over $X\times Y$, summing over all the $X$, i.e.~\emph{discarding} the state of $X$, gives the (marginal) distribution over $Y$:
\[
 \begin{tikzpicture}
	\begin{pgfonlayer}{nodelayer}
		\node [style=none] (7) at (0, 0.5) {};
		\node [style=none] (10) at (1.5, 0.5) {};
		\node [style=none] (14) at (0, -0.5) {};
		\node [style=none] (15) at (2, -0.5) {};
		\node [style=none] (16) at (2.5, -0.5) {$Y$};
		\node [style=horiz state] (17) at (0, 0) {$p$};
		\node [style=bn] (18) at (1.5, 0.5) {};
	\end{pgfonlayer}
	\begin{pgfonlayer}{edgelayer}
		\draw (7.center) to (10.center);
		\draw (14.center) to (15.center);
	\end{pgfonlayer}
\end{tikzpicture}
 \qquad\qquad
 \sum_{x} p(x,y) = p_Y(y).
\]
More on this in \Cref{divjm}.

These copy and discard maps are required to satisfy the following conditions, called \emph{commutative comonoid axioms}: first of all, copying and then discarding one of the copies is the same as doing nothing:
\[
 \begin{tikzpicture}
	\begin{pgfonlayer}{nodelayer}
		\node [style=none] (0) at (-8, 0) {};
		\node [style=none] (1) at (-8.5, 0) {$X$};
		\node [style=none] (2) at (-6, 1) {};
		\node [style=none] (3) at (-6, -1) {};
		\node [style=none] (4) at (-5.5, -1) {$X$};
		\node [style=bn] (5) at (-7, 0) {};
		\node [style=none] (6) at (6, 0) {};
		\node [style=none] (7) at (5.5, 0) {$X$};
		\node [style=none] (8) at (8, 1) {};
		\node [style=none] (9) at (8, -1) {};
		\node [style=none] (10) at (8.5, 1) {$X$};
		\node [style=bn] (11) at (7, 0) {};
		\node [style=bn] (12) at (-6, 1) {};
		\node [style=none] (13) at (-3.5, 0) {$=$};
		\node [style=none] (14) at (-1.5, 0) {$X$};
		\node [style=none] (15) at (1.5, 0) {$X$};
		\node [style=none] (16) at (-1, 0) {};
		\node [style=none] (17) at (1, 0) {};
		\node [style=none] (18) at (3.5, 0) {$=$};
		\node [style=bn] (19) at (8, -1) {};
	\end{pgfonlayer}
	\begin{pgfonlayer}{edgelayer}
		\draw (0.center) to (5);
		\draw [in=180, out=75] (5) to (2.center);
		\draw [in=180, out=-75] (5) to (3.center);
		\draw (6.center) to (11);
		\draw [in=180, out=75] (11) to (8.center);
		\draw [in=180, out=-75] (11) to (9.center);
		\draw (16.center) to (17.center);
	\end{pgfonlayer}
\end{tikzpicture}
\]
Second, copying the first copy has the same effect as copying the second copy (one just has three copies):
\[
 \begin{tikzpicture}
	\begin{pgfonlayer}{nodelayer}
		\node [style=none] (0) at (-6, -0.5) {};
		\node [style=none] (1) at (-6.5, -0.5) {$X$};
		\node [style=none] (2) at (-3.5, 1) {};
		\node [style=none] (3) at (-3.5, -2) {};
		\node [style=bn] (4) at (-5, -0.5) {};
		\node [style=bn] (5) at (-3.5, 1) {};
		\node [style=none] (6) at (-2.5, 2) {};
		\node [style=none] (7) at (-2.5, 0) {};
		\node [style=none] (8) at (2.5, 0.5) {};
		\node [style=none] (9) at (2, 0.5) {$X$};
		\node [style=none] (10) at (5, 2) {};
		\node [style=none] (11) at (5, -1) {};
		\node [style=bn] (12) at (3.5, 0.5) {};
		\node [style=bn] (13) at (5, -1) {};
		\node [style=none] (14) at (6, 0) {};
		\node [style=none] (15) at (6, -2) {};
		\node [style=none] (16) at (-2, 2) {$X$};
		\node [style=none] (17) at (-2, 0) {$X$};
		\node [style=none] (18) at (-2, -2) {$X$};
		\node [style=none] (19) at (6.5, 2) {$X$};
		\node [style=none] (20) at (6.5, 0) {$X$};
		\node [style=none] (21) at (6.5, -2) {$X$};
		\node [style=none] (22) at (0, 0) {$=$};
		\node [style=none] (23) at (-2.5, -2) {};
		\node [style=none] (24) at (6, 2) {};
	\end{pgfonlayer}
	\begin{pgfonlayer}{edgelayer}
		\draw (0.center) to (4);
		\draw [in=180, out=75] (4) to (2.center);
		\draw [in=180, out=-75] (4) to (3.center);
		\draw [in=75, out=-180] (6.center) to (5);
		\draw [in=-75, out=-180] (7.center) to (5);
		\draw (8.center) to (12);
		\draw [in=180, out=75] (12) to (10.center);
		\draw [in=180, out=-75] (12) to (11.center);
		\draw [in=75, out=-180] (14.center) to (13);
		\draw [in=180, out=-75] (13) to (15.center);
		\draw (23.center) to (3.center);
		\draw (24.center) to (10.center);
	\end{pgfonlayer}
\end{tikzpicture}
\]
Lastly, switching the two copies has no effect:
\[
 \begin{tikzpicture}
	\begin{pgfonlayer}{nodelayer}
		\node [style=none] (45) at (-6, 0) {};
		\node [style=none] (46) at (-6.5, 0) {$X$};
		\node [style=none] (47) at (-4, 1) {};
		\node [style=none] (48) at (-4, -1) {};
		\node [style=bn] (49) at (-5, 0) {};
		\node [style=none] (50) at (-2, 1) {};
		\node [style=none] (51) at (-2, -1) {};
		\node [style=none] (52) at (-1.5, -1) {$X$};
		\node [style=none] (53) at (-1.5, 1) {$X$};
		\node [style=none] (54) at (0, 0) {$=$};
		\node [style=none] (55) at (2, 0) {};
		\node [style=none] (56) at (1.5, 0) {$X$};
		\node [style=none] (57) at (4, 1) {};
		\node [style=none] (58) at (4, -1) {};
		\node [style=bn] (59) at (3, 0) {};
		\node [style=none] (60) at (5.5, 1) {$X$};
		\node [style=none] (61) at (5.5, -1) {$X$};
		\node [style=none] (62) at (5, 1) {};
		\node [style=none] (63) at (5, -1) {};
	\end{pgfonlayer}
	\begin{pgfonlayer}{edgelayer}
		\draw (45.center) to (49);
		\draw [in=180, out=75] (49) to (47.center);
		\draw [in=180, out=-75] (49) to (48.center);
		\draw [in=0, out=180] (50.center) to (48.center);
		\draw (55.center) to (59);
		\draw [in=180, out=75] (59) to (57.center);
		\draw [in=180, out=-75] (59) to (58.center);
		\draw (63.center) to (58.center);
		\draw (62.center) to (57.center);
		\draw [style=over arrow, in=-180, out=0] (47.center) to (51.center);
	\end{pgfonlayer}
\end{tikzpicture}
\]
Moreover, we require these copy and discard maps to be compatible with the tensor product.

Note that a version of this copy and discard structure is implicitly used whenever information is manipulated. For example, when we have channels $f,g:X\to Y$ and write expressions such as 
\[
 \big( f(x), g(x) \big) ,
\]
feeding the same value $x$ in both functions (and not, for example, $x$ and $x'$) we are implicitly using the copy map, as follows.
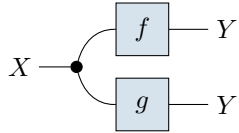
\[
 \begin{tikzpicture}
	\begin{pgfonlayer}{nodelayer}
		\node [style=none] (0) at (-1, 0) {};
		\node [style=none] (1) at (-1.5, 0) {$X$};
		\node [style=none] (2) at (1, 1) {};
		\node [style=none] (3) at (1, -1) {};
		\node [style=bn] (6) at (0, 0) {};
		\node [style=medium box] (7) at (1.75, 1) {$f$};
		\node [style=medium box] (8) at (1.75, -1) {$g$};
		\node [style=none] (9) at (3.5, 1) {};
		\node [style=none] (10) at (3.5, -1) {};
		\node [style=none] (11) at (4, 1) {$Y$};
		\node [style=none] (12) at (4, -1) {$Y$};
	\end{pgfonlayer}
	\begin{pgfonlayer}{edgelayer}
		\draw (0.center) to (6);
		\draw [in=180, out=75] (6) to (2.center);
		\draw [in=180, out=-75] (6) to (3.center);
		\draw (2.center) to (9.center);
		\draw (3.center) to (10.center);
	\end{pgfonlayer}
\end{tikzpicture}
\]
Similarly, whenever we have a source $p$ on $X$, we can view it as a constant (noisy) channel $A\to X$ with an input $A$ which is not really used for processing. This can be expressed using the discard map as follows.
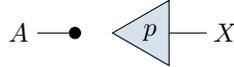
\[
 \begin{tikzpicture}
	\begin{pgfonlayer}{nodelayer}
		\node [style=none] (0) at (-1, 0) {};
		\node [style=none] (1) at (-1.5, 0) {$A$};
		\node [style=bn] (6) at (0, 0) {};
		\node [style=none] (7) at (2, 0) {};
		\node [style=none] (8) at (3.5, 0) {};
		\node [style=none] (9) at (4, 0) {$X$};
		\node [style=horiz state] (10) at (2, 0) {$p$};
	\end{pgfonlayer}
	\begin{pgfonlayer}{edgelayer}
		\draw (0.center) to (6);
		\draw (7.center) to (8.center);
	\end{pgfonlayer}
\end{tikzpicture}
\]
More generally, any channel $f:X\to Y$ can also be seen as a channel $X\otimes A\to Y$ which does not use the input $A$, as follows.
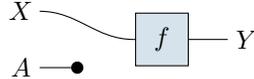
\[
 \begin{tikzpicture}
	\begin{pgfonlayer}{nodelayer}
		\node [style=none] (0) at (-1, -0.75) {};
		\node [style=none] (1) at (-1.5, -0.75) {$A$};
		\node [style=bn] (6) at (0, -0.75) {};
		\node [style=none] (11) at (-1, 0.75) {};
		\node [style=none] (12) at (1.5, 0) {};
		\node [style=none] (13) at (-1.5, 0.75) {$X$};
		\node [style=none] (14) at (4.5, 0) {$Y$};
		\node [style=medium box] (15) at (2.25, 0) {$f$};
		\node [style=none] (16) at (4, 0) {};
	\end{pgfonlayer}
	\begin{pgfonlayer}{edgelayer}
		\draw (0.center) to (6);
		\draw [in=180, out=0] (11.center) to (12.center);
		\draw (12.center) to (16.center);
	\end{pgfonlayer}
\end{tikzpicture}
\]
 
The last property that we require in a Markov category is \emph{normalization} or \emph{counitality}: applying a morphism $f$ and discarding its output is the same as discarding the input from the start.
\[
 \begin{tikzpicture}
	\begin{pgfonlayer}{nodelayer}
		\node [style=medium box] (0) at (-4, 0) {$f$};
		\node [style=none] (1) at (-5.5, 0) {};
		\node [style=none] (2) at (-1.5, 0) {};
		\node [style=none] (3) at (-6, 0) {$X$};
		\node [style=none] (4) at (-2.25, 0.5) {$Y$};
		\node [style=bn] (5) at (-1.5, 0) {};
		\node [style=none] (6) at (0, 0) {$=$};
		\node [style=none] (7) at (2, 0) {};
		\node [style=none] (8) at (1.5, 0) {$X$};
		\node [style=bn] (9) at (4, 0) {};
	\end{pgfonlayer}
	\begin{pgfonlayer}{edgelayer}
		\draw (2.center) to (1.center);
		\draw (9) to (7.center);
	\end{pgfonlayer}
\end{tikzpicture}
\]
In $\cat{FinStoch}$, this is exactly the condition that the sum of each column of a stochastic matrix is one, i.e.~that transition probabilities are normalized. 

These structures and properties are what is needed to form a Markov category. 
For reference, here is the rigorous, concise definition.

\begin{definition}\label{defmarkov}
 A \emph{Markov category} is a symmetric monoidal category $(\cat{C},\otimes,I)$ together with a chosen commutative comonoid structure for each object $X$, which is compatible with tensor products, and for which all morphisms are counital.
\end{definition}

The counitality or normalization condition is sometimes dropped, and instead of a Markov category one talks about a \emph{garbage-share (GS)}~\cite{gadducci-thesis,fritz2022free} or \emph{copy-discard (CD) category}~\cite{chojacobs2019strings}. 

For more information on the theory of Markov categories we refer to the original source \cite{fritz2019synthetic}, and to the other material cited in the introduction.
Note that in most other articles the graphical calculus is written vertically, from bottom to top, instead of from left to right.

\section{Divergences on Markov categories}\label{divmcat}

\begin{definition}\label{defdiv} 
 A \emph{divergence} or \emph{statistical distance} on a set $X$ is a function 
 \[
  \begin{tikzcd}[row sep=0]
   X\times X \ar{r}{D} & {[0,\infty]} \\
   (x,y) \ar[mapsto]{r} & D(x\et y)
  \end{tikzcd}
 \]
 such that $D(x\et x)=0$.
 
 We call the pair $(X,D)$ a \emph{divergence space}.
 
 We call the divergence $D$ \emph{strict} if $D(x\et y)=0$ implies $x=y$.
\end{definition}

Every metric space is a strict divergence space. Divergences, however, are not required to be symmetric, nor to satisfy the triangle inequality (and in our convention, infinity is allowed). Still, the same intuition can help.

\begin{remark}
 For the readers who find enriched category theory helpful, a divergence is to a (Lawvere) metric as a reflexive multigraph is to a category (or as a reflexive relation is to a preorder).
 The category of divergence spaces, and its usage as an enriching category, is explained in more detail in \Cref{Div}.
\end{remark}
 
Here is an example of a non-metric divergence. (First, a convention.) 

\begin{convention}\label{logzero}
 In expressions such as $x \ln \dfrac{y}{z}$, we set $ 0 \ln \dfrac{0}{x} = 0 \ln \dfrac{0}{0} = 0$, and $x \ln \dfrac{x}{0} = \infty$ for $x\ne 0$.
 In particular, $0 \ln 0 = 0$. 
\end{convention}

\begin{definition}\label{defKL}
 Let $X$ be a finite set, and let $p$ and $q$ be probability distributions on $X$. The \emph{relative entropy}, or \emph{Kullback-Leibler} (\emph{KL}) \emph{divergence}, between $p$ and $q$ is the quantity
 \[
 D_{KL}(p\et q) \coloneqq \sum_{x\in X} p(x) \ln \dfrac{p(x)}{q(x)} ,
 \]
 using \Cref{logzero}.
\end{definition}

The space $PX$ of probability distributions on a finite set $X$, together with the relative entropy, forms a divergence space. Note that $D_{KL}(p\et q)$ is finite if and only if the support of $p$ is contained in the support of $q$, or in measure-theoretic terms, $p$ is absolutely continuous with respect to $q$. 

We now consider Markov categories where sources and channels are equipped with a family of chosen, compatible divergences.

\begin{definition}\label{defdivmarkov}
 A \emph{divergence} on a Markov category $\cat{C}$ amounts to the following data.
 \begin{itemize}
  \item For each pair of objects $X$ and $Y$, a divergence $D_{X,Y}$ on the set of morphisms $X\to Y$, or more briefly just $D$;
 \end{itemize}
 such that 
 \begin{itemize}
  \item The composition of morphisms in the following form
  \[
   \begin{tikzcd}
    X \ar[bend left]{r}{f} \ar[bend right]{r}[swap]{f'}
    & Y \ar[bend left]{r}{g} \ar[bend right]{r}[swap]{g'}
    & Z
   \end{tikzcd}
  \]
 satisfies the following inequality,
  \begin{equation}\label{enrichmentcondition}
   D(g\circ f\et g'\circ f') \le D(f\et f') + D(g\et g') ;
  \end{equation}
  \item The tensor product of morphisms in the following form
  \[
   \begin{tikzcd}[row sep=0]
    X\otimes A \ar[bend left]{r}{f\otimes h} \ar[bend right]{r}[swap]{f'\otimes h'} & Y\otimes B
   \end{tikzcd}
  \]
  satisfies the following inequality,
  \begin{equation}\label{tensortriangle}
   D\big( (f\otimes h)\et  (f'\otimes h') \big) \le D(f\et  f') + D(h\et  h') .
  \end{equation}
 \end{itemize}
\end{definition}

An interpretation of inequalities \eqref{enrichmentcondition} and \eqref{tensortriangle} is that one can bound the divergence between complex configurations of sources and channels, obtained through sequential and parallel composition, in terms of their simpler components. 
For example, the distance or divergence between the two systems depicted below
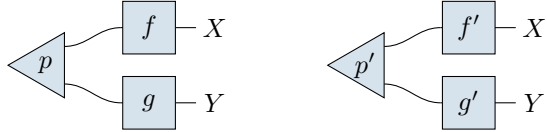
\[
\begin{tikzpicture}
	\begin{pgfonlayer}{nodelayer}
		\node [style=medium box] (1) at (-3.75, -1) {$g$};
		\node [style=none] (5) at (-4.5, 1) {};
		\node [style=none] (6) at (-4.5, -1) {};
		\node [style=none] (7) at (-2.5, -1) {};
		\node [style=none] (8) at (-2.5, 1) {};
		\node [style=none] (9) at (-2, 1) {$X$};
		\node [style=none] (10) at (-2, -1) {$Y$};
		\node [style=horiz state] (11) at (-6.5, 0) {$p$};
		\node [style=medium box] (12) at (-3.75, 1) {$f$};
		\node [style=none] (13) at (-6, 0.5) {};
		\node [style=none] (14) at (-6, -0.5) {};
		\node [style=medium box] (15) at (4.75, -1) {$g'$};
		\node [style=none] (16) at (4, 1) {};
		\node [style=none] (17) at (4, -1) {};
		\node [style=none] (18) at (6, -1) {};
		\node [style=none] (19) at (6, 1) {};
		\node [style=none] (20) at (6.5, 1) {$X$};
		\node [style=none] (21) at (6.5, -1) {$Y$};
		\node [style=horiz state] (22) at (2, 0) {$p'$};
		\node [style=medium box] (23) at (4.75, 1) {$f'$};
		\node [style=none] (24) at (2.5, 0.5) {};
		\node [style=none] (25) at (2.5, -0.5) {};
	\end{pgfonlayer}
	\begin{pgfonlayer}{edgelayer}
		\draw (7.center) to (6.center);
		\draw (5.center) to (8.center);
		\draw [in=-180, out=0] (13.center) to (5.center);
		\draw [in=180, out=0] (14.center) to (6.center);
		\draw (18.center) to (17.center);
		\draw (16.center) to (19.center);
		\draw [in=-180, out=0] (24.center) to (16.center);
		\draw [in=180, out=0] (25.center) to (17.center);
	\end{pgfonlayer}
\end{tikzpicture}
\]
is bounded by $D(p,p') + D(f,f') + D(g,g')$. More generally, for any two string diagrams of any configuration, the distance or divergence between the resulting constructions will always be bounded by the divergence between the basic building blocks.\footnote{If the configurations do not correspond exactly, because their ``wiring'' is different, one still has a bound for each way of partially pattern-matching them.}
In particular, setting for example $p=p'$ and $f=f'$ but not $g=g'$ in the diagrams above, we see that the divergence between $g$ and $g'$ is not increased by pre-processing, post-processing or parallel processing $g$ and $g'$ with the same subsystem. For the case of post-processing, this gives Shannon-like data processing inequalities, see \Cref{qconvdproc}.

Readers familiar with enriched category theory might recognize an enrichment in \Cref{defdivmarkov}. This is indeed the case, and for the interested readers, more details are given in \Cref{Div}. 
Note also that in the definition we are using the monoidal structure of $\cat{C}$, but not the Markov structure. The latter will however be used to give a simpler description, in \Cref{equivenrich}.

In this work we will focus on divergences on the category of finite alphabet channels ($\cat{FinStoch}$) and on the category of possibly infinite and continuous Markov kernels ($\cat{Stoch}$). If we have a divergence on the space $PY$ of probability measures over $Y$, we can define a divergence between two channels $f,g:X\to Y$ by taking the supremum over the inputs,
\begin{equation}\label{Dtocat}
 D(f\et g) \coloneqq \sup_{x\in X} D_Y \big( f_x\et g_x \big) ,
\end{equation}
recalling that $f_x$ is the probability distribution on $Y$ given by mapping a measurable set $T\subseteq Y$ to $f(T|x)$ (or just $y\to f(y|x)$ in the finite case), and analogously for $g_x$.
If $Y$ is finite, we can take the maximum instead of the supremum.

For divergences obtained in this form, the conditions of \Cref{defdivmarkov} can be checked in a particularly simple form. 
Let's call a \emph{family of divergences} a choice of a divergence $D_X$ on $PX$ for each measurable set $X$ (or finite set $X$, in the finite case). For example, the Kullback-Leibler divergence is one such family.

\begin{proposition}\label{simpler}
 A family of divergences $\{D_X\}$ on measurable sets (resp.~finite sets) gives a divergence on $\cat{Stoch}$ (resp.~$\cat{FinStoch}$) if and only if 
 \begin{equation}\label{simplerenrich}
  D_Y\big( f\circ p\et f'\circ p' \big) \le  D_X\big(p\et p' \big) + \sup_{x\in X} D_Y\big( f_x\et  f'_x \big) 
 \end{equation}
 for all probability distributions $p,p'$ on $X$ and kernels $f,f':X\to Y$,
 and 
 \begin{equation}\label{simplertens}
  D_{X\otimes A}(p\otimes q\et  p'\otimes q') \le D_X(p\et p') + D_A(q\et q') 
 \end{equation}
 for each probability distribution $p,p'$ on $X$ and $q,q'$ on $A$.
\end{proposition}

\begin{proof}
 Let $\{D_X\}$ be a family of divergences, and suppose that \eqref{simplerenrich} and \eqref{simplertens} are satisfied for all distributions and kernels. Condition \eqref{enrichmentcondition}, using \eqref{Dtocat}, reads as follows.
 \[
  \sup_{x\in X} D_Z\big( (g\circ f)_x\et  (g'\circ f')_x \big) \le \sup_{x\in X} D_Y( f_x\et  f'_x ) + \sup_{y\in Y} D_Z( g_y\et  g'_y ) . 
 \]
 The inequality above holds in particular if it holds for every $x$ individually, without taking the supremum. We are left with
 \[
  D_Z\big( (g\circ f)_x\et  (g'\circ f')_x \big) \le \sup_{x\in X} D_Y( f_x\et  f'_x ) + \sup_{y\in Y} D_Z( g_y\et  g'_y )
 \]
 for every $x$, which up to renaming ($f_x$ to $p$, $g$ to $f$, etc.) is \eqref{simplerenrich}.
 
 Similarly, \eqref{tensortriangle} reads as follows.
 \[
  \sup_{x\in X, a\in A} D_{Y\otimes B} \big( f_x\otimes h_a\et f'_x\otimes h'_a \big) 
  \le \sup_{x\in X} D_Y\big( f_x\et  f'_x \big)  + \sup_{a\in A} D_B\big( h_a\et  h'_a \big) .
 \]
 Once again, the inequality holds in particular if it holds for every $x$ and $a$ individually, without taking the supremum, i.e.~if for all $x$ and $a$,
 \[
  D_{Y\otimes B} \big( f_x\otimes h_a\et f'_x\otimes h'_a \big) \le D_Y\big( f_x\et  f'_x \big) + D_B\big( h_a\et  h'_a \big)
 \]
 Again, up to renaming, it is sufficient to prove \eqref{simplertens}.
 
 The converse statement is obtained by taking $X$ and $A$ to be one-point spaces (and suitably renaming). 
\end{proof}

\subsection{Data processing and other inequalities}\label{qconvdproc}

In \Cref{defdivmarkov}, we have seen that in order to have a divergence on our Markov category via \eqref{Dtocat}, for each diagram in the following form,
   \[
   \begin{tikzcd}
    X \ar[bend left]{r}{f} \ar[bend right]{r}[swap]{f'}
    & Y\ar[bend left]{r}{g} \ar[bend right]{r}[swap]{g'}
    & Z
   \end{tikzcd}
   \]
the following inequality needs to be satisfied,
  \[
   D(g\circ f\et g'\circ f') \le D(f\et f') + D(g\et g') .
  \]
  
Let's look at the information-theoretic meaning of this inequality. First of all, setting $g=g'$,
   \[
   \begin{tikzcd}
    X \ar[bend left]{r}{f} \ar[bend right]{r}[swap]{f'}
    & Y \ar{r}{g} 
    & Z
   \end{tikzcd}
   \]
we get 
\begin{equation}\label{Ddataprocgen}
 D(g\circ f\et g\circ f') \le D(f\et f') ,
\end{equation}
and for sources,
\begin{equation}\label{Ddataproc}
 D(g\circ p\et g\circ p') \le D(p\et p') .
\end{equation}
We can interpret this condition as a \emph{data processing inequality}: the idea is that if we process $X$ with the channel $g$, then we might lose some distinctions, and hence the divergence between the two inputs is decreased by processing.
This is particularly important when $g$ is a deterministic function, but the condition holds more in general, also when $g$ is a kernel with randomness.
In terms of random variables, the condition reads
\[
 D\big(g(X) \et g(Y) \big) \le D(X \et Y).
\]
See for example \cite[Example~2]{renyi-div} for more on this idea.
 
If instead in \eqref{enrichmentcondition} we set $f=f'$,
\[
   \begin{tikzcd}
    X \ar{r}{f}
    & Y \ar[bend left]{r}{g} \ar[bend right]{r}[swap]{g'}
    & Z
   \end{tikzcd}
   \]
the corresponding condition that we get is that
  \begin{equation}\label{Dquasiconv}
   D(g\circ f\et g'\circ f) \le D(g\et g') ,
  \end{equation}
which is related, but not equivalent, to quasi-convexity of $D$. 

One might now ask, if a divergence satisfies both \eqref{Ddataprocgen} and \eqref{Dquasiconv} for all channels, does it also satisfy \eqref{simplerenrich}? A partial answer is, \emph{it does when the divergence is a metric}. 
Indeed, let $X$, $Y$ and $Z$ be objects of a category $\cat{C}$, and consider divergences $D$ on the hom-sets $\cat{C}(X,Y)$, $\cat{C}(Y,Z)$, and $\cat{C}(X,Z)$ (for example given by taking the supremum of divergences on $Y$ and $Z$, as in \eqref{Dtocat}). 
 Suppose moreover that the divergence on $\cat{C}(X,Z)$ satisfies a metric triangle inequality.
 If given morphisms as follows,
 \[
  \begin{tikzcd}
    X \ar[bend left]{r}{f} \ar[bend right]{r}[swap]{f'}
    & Y \ar[bend left]{r}{g} \ar[bend right]{r}[swap]{g'}
    & Z
   \end{tikzcd}
 \]
 we have that
 \[
  D(g\circ f\et g\circ f') \le D(f\et f')
 \]
 and 
 \[
  D(g\circ f'\et g'\circ f') \le D(g\et g') ,
 \]
 then since $D$ on $\cat{C}(X,Z)$ satisfies a triangle inequality,
 \begin{align*}
  D(g\circ f\et g'\circ f') &\le D(g\circ f\et g\circ f') + D(g\circ f'\et g'\circ f') \\
  &\le D(g\et g') + D(f\et f') .
 \end{align*}
(In general, if there is no metric triangle inequality, this argument does not work.)
 
Let's try to interpret this diagrammatically. 
The condition on composition \eqref{enrichmentcondition} gives a sort of ``horizontal'' or ``sequential'' triangle inequality, which we could intuitively draw as follows.
\begin{equation}\label{enrichmentcondition2}
 \begin{tikzcd}
    X \ar[bend left]{r}{f} \ar[bend right]{r}[swap]{f'}
    & Y \ar[bend left]{r}{g} \ar[bend right]{r}[swap]{g'}
    & Z
   \end{tikzcd}
   \qquad \le \qquad
    \begin{tikzcd}
    X \ar[bend left]{r}{f} \ar[bend right]{r}[swap]{f'}
    & Y 
   \end{tikzcd}
   \quad+\quad
    \begin{tikzcd}
    Y \ar[bend left]{r}{g} \ar[bend right]{r}[swap]{g'}
    & Z
   \end{tikzcd}
\end{equation}
The idea is that we can bound the divergence between \emph{sequential} compositions of processes in terms of their components. 
Instead, a metric triangle inequality is ``vertical'', we could write it as follows,
\[
 \begin{tikzcd}
  X \ar[bend left=45]{rr}{a} \ar{rr}{b} \ar[bend right=45]{rr}[swap]{c} && Y
 \end{tikzcd}
 \qquad\le\qquad
 \begin{tikzcd}[row sep=tiny, column sep=small]
  X \ar[bend left=45]{rr}{a} \ar{rr}{b} && Y \\
  & + \\
  X \ar{rr}{b} \ar[bend right=45]{rr}[swap]{c} && Y
 \end{tikzcd}
 \]
and in general it may or may not hold regardless of whether the ``horizontal'' inequality \eqref{enrichmentcondition2} holds.
However, if this metric triangle inequality holds, we can decompose the sequential composition as follows,
\[
 \begin{tikzcd}
    X \ar[bend left]{r}{f} \ar[bend right]{r}[swap]{f'}
    & Y \ar[bend left]{r}{g} \ar[bend right]{r}[swap]{g'}
    & Z
   \end{tikzcd}
   \qquad \le \qquad
    \begin{tikzcd}[row sep=small]
    X \ar[bend left]{r}{f} \ar[bend right]{r}[swap]{f'}
    & Y \ar[bend left]{r}{g}
    & Z \\
    & + \\
    X \ar[bend right]{r}[swap]{f'}
    & Y \ar[bend left]{r}{g} \ar[bend right]{r}[swap]{g'}
    & Z
   \end{tikzcd}
\]
and if we have an inequality for both terms on the right-hand side (the data processing and quasi-convexity conditions, intuitively),
\[
 \begin{tikzcd}
    X \ar[bend left]{r}{f} \ar[bend right]{r}[swap]{f'}
    & Y \ar[bend left]{r}{g}
    & Z 
   \end{tikzcd}
   \quad \le \quad
   \begin{tikzcd}
    X \ar[bend left]{r}{f} \ar[bend right]{r}[swap]{f'}
    & Y
   \end{tikzcd}
\]
and
\[
    \begin{tikzcd}
    X \ar[bend right]{r}[swap]{f'}
    & Y \ar[bend left]{r}{g} \ar[bend right]{r}[swap]{g'}
    & Z
   \end{tikzcd}
   \quad\le\quad
   \begin{tikzcd}
    Y \ar[bend left]{r}{g} \ar[bend right]{r}[swap]{g'}
    & Z ,
   \end{tikzcd}
\]
we can effectively obtain the sequential inequality \eqref{enrichmentcondition2}.

The tensor product condition \eqref{tensortriangle} is again a sort of triangle inequality, but in yet again a different ``direction'', namely in the direction of parallel processing. It says that we can bound the divergence between \emph{parallel} sets of independent processes in terms of their components.
Again, in general this is independent of whether the divergences satisfy a metric triangle inequality or not.

\subsection{Characterization in terms of joints and marginals}\label{divjm}

We have an equivalent characterization of Markov categories with divergences, which focuses on \emph{joint} and \emph{marginal} morphisms and distributions, rather than sequential composition. 
First of all, recall that given finitely supported probability distributions $p$ on $X$ and $q$ on $Y$, a \emph{joint distribution} of $p$ and $q$ (or of $X$ and $Y$, if we see them as random variables) is a probability distribution $r$ on $X\times Y$ such that 
\[
 \sum_{y\in Y} r(x,y) = p(x) \qquad\mbox{and}\qquad \sum_{x\in X} r(x,y) = q(y) .
\]
We call $p$ and $q$ the \emph{marginals} of $p$.
In a generic Markov category, given sources $p$ on $X$ and $q$ on $Y$, a \emph{joint source} of $p$ and $q$ is a source $r$ on $X\otimes Y$, 
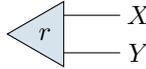
\[
 \begin{tikzpicture}
	\begin{pgfonlayer}{nodelayer}
		\node [style=horiz state] (0) at (-1, 0) {$r$};
		\node [style=none] (1) at (1, 0.5) {};
		\node [style=none] (2) at (1, -0.5) {};
		\node [style=none] (3) at (-1, 0.5) {};
		\node [style=none] (4) at (-1, -0.5) {};
		\node [style=none] (5) at (1.5, 0.5) {$X$};
		\node [style=none] (6) at (1.5, -0.5) {$Y$};
	\end{pgfonlayer}
	\begin{pgfonlayer}{edgelayer}
		\draw (4.center) to (2.center);
		\draw (3.center) to (1.center);
	\end{pgfonlayer}
\end{tikzpicture}
\] 
such that the following holds,
\[
 \begin{tikzpicture}
	\begin{pgfonlayer}{nodelayer}
		\node [style=horiz state] (7) at (-11, 0) {$r$};
		\node [style=none] (8) at (-9, 0.5) {};
		\node [style=none] (9) at (-9.5, -0.5) {};
		\node [style=none] (10) at (-11, 0.5) {};
		\node [style=none] (11) at (-11, -0.5) {};
		\node [style=none] (12) at (-8.5, 0.5) {$X$};
		\node [style=bn] (13) at (-9.5, -0.5) {};
		\node [style=none] (14) at (-7.25, 0) {$=$};
		\node [style=horiz state] (15) at (-5, 0) {$p$};
		\node [style=none] (16) at (-3, 0) {};
		\node [style=none] (17) at (-2.5, 0) {$X$};
		\node [style=horiz state] (18) at (3, 0) {$r$};
		\node [style=none] (19) at (4.5, 0.5) {};
		\node [style=none] (20) at (5, -0.5) {};
		\node [style=none] (21) at (3, 0.5) {};
		\node [style=none] (22) at (3, -0.5) {};
		\node [style=none] (24) at (5.5, -0.5) {$Y$};
		\node [style=bn] (25) at (4.5, 0.5) {};
		\node [style=none] (26) at (6.75, 0) {$=$};
		\node [style=horiz state] (27) at (9, 0) {$q$};
		\node [style=none] (28) at (11, 0) {};
		\node [style=none] (29) at (11.5, 0) {$Y$};
		\node [style=none] (30) at (0, 0) {and};
	\end{pgfonlayer}
	\begin{pgfonlayer}{edgelayer}
		\draw (11.center) to (9.center);
		\draw (10.center) to (8.center);
		\draw (16.center) to (15);
		\draw (22.center) to (20.center);
		\draw (21.center) to (19.center);
		\draw (27) to (28.center);
	\end{pgfonlayer}
\end{tikzpicture}
\]
in formulas, $(\id_X\circ\discard_Y)\circ r = p$ and $(\discard_X\circ\id_Y)=q$. 
We call $p$ and $q$ the \emph{marginal sources} of $r$, and denote them by $p=r_X$ and $q=r_Y$.
Notice that the ``discard'' maps correspond exactly to the sums in the finite probability case: marginalizing, or integrating over $X$, is encoded by ``discarding'' our information about $X$.

More generally, if $p$ and $q$ depend on an additional parameter $A$, a \emph{joint morphism} of $p:A\to X$ and $q:A\to X$ is a morphism $r:A\to X\otimes Y$,
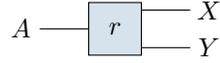
\[
 \begin{tikzpicture}
	\begin{pgfonlayer}{nodelayer}
		\node [style=none] (1) at (1, 0.5) {};
		\node [style=none] (2) at (1, -0.5) {};
		\node [style=none] (3) at (-1, 0.5) {};
		\node [style=none] (4) at (-1, -0.5) {};
		\node [style=none] (5) at (1.5, 0.5) {$X$};
		\node [style=none] (6) at (1.5, -0.5) {$Y$};
		\node [style=medium box] (7) at (-1, 0) {$r$};
		\node [style=none] (8) at (-3, 0) {};
		\node [style=none] (9) at (-3.5, 0) {$A$};
	\end{pgfonlayer}
	\begin{pgfonlayer}{edgelayer}
		\draw (4.center) to (2.center);
		\draw (3.center) to (1.center);
		\draw (8.center) to (7);
	\end{pgfonlayer}
\end{tikzpicture}
\] 
such that the following holds.
\[
 \begin{tikzpicture}
	\begin{pgfonlayer}{nodelayer}
		\node [style=none] (8) at (-8.5, 0.5) {};
		\node [style=none] (9) at (-8.75, -0.5) {};
		\node [style=none] (10) at (-10, 0.5) {};
		\node [style=none] (11) at (-10, -0.5) {};
		\node [style=none] (12) at (-8, 0.5) {$X$};
		\node [style=bn] (13) at (-8.75, -0.5) {};
		\node [style=none] (14) at (-6.75, 0) {$=$};
		\node [style=none] (16) at (-2.5, 0) {};
		\node [style=none] (17) at (-2, 0) {$X$};
		\node [style=none] (19) at (5, 0.5) {};
		\node [style=none] (20) at (5.25, -0.5) {};
		\node [style=none] (21) at (3.75, 0.5) {};
		\node [style=none] (22) at (3.75, -0.5) {};
		\node [style=none] (24) at (5.75, -0.5) {$Y$};
		\node [style=bn] (25) at (5, 0.5) {};
		\node [style=none] (26) at (7, 0) {$=$};
		\node [style=none] (28) at (11.25, 0) {};
		\node [style=none] (29) at (11.75, 0) {$Y$};
		\node [style=none] (30) at (0, 0) {and};
		\node [style=medium box] (31) at (3.75, 0) {$r$};
		\node [style=none] (32) at (2.5, 0) {};
		\node [style=none] (33) at (2, 0) {$A$};
		\node [style=none] (34) at (8.75, 0) {};
		\node [style=none] (35) at (8.25, 0) {$A$};
		\node [style=medium box] (36) at (-10, 0) {$r$};
		\node [style=none] (37) at (-11.25, 0) {};
		\node [style=none] (38) at (-11.75, 0) {$A$};
		\node [style=medium box] (39) at (10, 0) {$q$};
		\node [style=medium box] (40) at (-3.75, 0) {$p$};
		\node [style=none] (41) at (-5, 0) {};
		\node [style=none] (42) at (-5.5, 0) {$A$};
	\end{pgfonlayer}
	\begin{pgfonlayer}{edgelayer}
		\draw (11.center) to (9.center);
		\draw (10.center) to (8.center);
		\draw (22.center) to (20.center);
		\draw (21.center) to (19.center);
		\draw (32.center) to (31);
		\draw (37.center) to (36);
		\draw (34.center) to (28.center);
		\draw (41.center) to (16.center);
	\end{pgfonlayer}
\end{tikzpicture}
\]
(In formulas we have again that $(\id_X\circ\discard_Y)\circ r = p$ and $(\discard_X\circ\id_Y)=q$.)
We call $p$ and $q$ the \emph{marginal morphisms} of $r$, and again write $p=r_X$ and $q=r_Y$.
This corresponds, for finite probability distributions, to the conditions
\[
 \sum_{y\in Y} r(x,y|a) = p(x|a) \qquad\mbox{and}\qquad \sum_{x\in X} r(x,y|a) = q(y|a) .
\]

Recall also that from a finitely supported probability measure $p$ on $X$ and a stochastic matrix $f:X\to Y$ we can form the \emph{joint probability} $fp$ on $X\times Y$ by 
\begin{equation}\label{jointformula}
 fp\,(x,y) \coloneqq p(x) \, f(y|x) .
\end{equation}
(Note that other authors write this differently, for example $f\circ p$, while for us $f\circ p$ denotes the resulting distribution on $Y$.)
In a generic Markov category $\cat{C}$, given a source $p$ on $X$ and a morphism $f:X\to Y$, we can form the \emph{joint source} $fp$ on $X\otimes Y$ as follows.
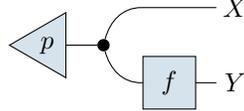
\[ 
 \begin{tikzpicture}
	\begin{pgfonlayer}{nodelayer}
		\node [style=medium box] (1) at (1.25, -1) {$f$};
		\node [style=bn] (2) at (-0.5, 0) {};
		\node [style=none] (3) at (-2, 0) {};
		\node [style=none] (5) at (0.5, 1) {};
		\node [style=none] (6) at (0.5, -1) {};
		\node [style=none] (7) at (2.5, -1) {};
		\node [style=none] (8) at (2.5, 1) {};
		\node [style=none] (9) at (3, 1) {$X$};
		\node [style=none] (10) at (3, -1) {$Y$};
		\node [style=horiz state] (11) at (-2, 0) {$p$};
	\end{pgfonlayer}
	\begin{pgfonlayer}{edgelayer}
		\draw (3.center) to (2);
		\draw [in=75, out=-180] (5.center) to (2);
		\draw [in=-75, out=180] (6.center) to (2);
		\draw (7.center) to (6.center);
		\draw (5.center) to (8.center);
	\end{pgfonlayer}
\end{tikzpicture}
\]
(In formulas, $fp = (\id_X\otimes f)\circ\cop_X\circ p$.)
The copy map is used, analogously to how the (same) value $x$ appears twice in equation \ref{jointformula}.
The marginals of this joint source, just like in the finite probability case, are $(fp)_X=p$ and $(fp)_Y=f\circ p$:
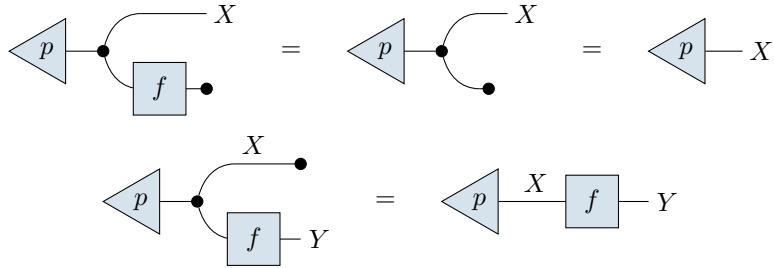
\[
 \begin{tikzpicture}
	\begin{pgfonlayer}{nodelayer}
		\node [style=medium box] (1) at (-4, -3.5) {$f$};
		\node [style=bn] (2) at (-5.5, -2.5) {};
		\node [style=none] (3) at (-7, -2.5) {};
		\node [style=none] (5) at (-4.5, -1.5) {};
		\node [style=none] (6) at (-4.5, -3.5) {};
		\node [style=none] (7) at (-2.75, -3.5) {};
		\node [style=none] (8) at (-2.75, -1.5) {};
		\node [style=none] (10) at (-2.25, -3.5) {$Y$};
		\node [style=horiz state] (11) at (-7, -2.5) {$p$};
		\node [style=medium box] (12) at (-6.5, 0.5) {$f$};
		\node [style=bn] (13) at (-8, 1.5) {};
		\node [style=none] (14) at (-9.5, 1.5) {};
		\node [style=none] (15) at (-7, 2.5) {};
		\node [style=none] (16) at (-7, 0.5) {};
		\node [style=none] (17) at (-5.25, 0.5) {};
		\node [style=none] (18) at (-5.25, 2.5) {};
		\node [style=none] (19) at (-4.75, 2.5) {$X$};
		\node [style=horiz state] (21) at (-9.5, 1.5) {$p$};
		\node [style=bn] (22) at (-2.75, -1.5) {};
		\node [style=none] (23) at (-0.5, -2.5) {$=$};
		\node [style=horiz state] (24) at (2, -2.5) {$p$};
		\node [style=none] (25) at (6.5, -2.5) {};
		\node [style=none] (26) at (7, -2.5) {$Y$};
		\node [style=medium box] (27) at (5, -2.5) {$f$};
		\node [style=bn] (29) at (-5.25, 0.5) {};
		\node [style=none] (30) at (-3, 1.5) {$=$};
		\node [style=bn] (32) at (1, 1.5) {};
		\node [style=none] (33) at (-0.5, 1.5) {};
		\node [style=none] (34) at (2, 2.5) {};
		\node [style=none] (35) at (2, 0.5) {};
		\node [style=none] (36) at (2.25, 0.5) {};
		\node [style=none] (37) at (2.75, 2.5) {};
		\node [style=none] (38) at (3.25, 2.5) {$X$};
		\node [style=horiz state] (39) at (-0.5, 1.5) {$p$};
		\node [style=bn] (40) at (2.25, 0.5) {};
		\node [style=none] (41) at (5, 1.5) {$=$};
		\node [style=horiz state] (42) at (7.5, 1.5) {$p$};
		\node [style=none] (43) at (9, 1.5) {};
		\node [style=none] (44) at (9.5, 1.5) {$X$};
		\node [style=none] (45) at (-4, -1) {$X$};
		\node [style=none] (46) at (3.5, -2) {$X$};
	\end{pgfonlayer}
	\begin{pgfonlayer}{edgelayer}
		\draw (3.center) to (2);
		\draw [in=75, out=-180] (5.center) to (2);
		\draw [in=-75, out=180] (6.center) to (2);
		\draw (7.center) to (6.center);
		\draw (5.center) to (8.center);
		\draw (14.center) to (13);
		\draw [in=75, out=-180] (15.center) to (13);
		\draw [in=-75, out=180] (16.center) to (13);
		\draw (17.center) to (16.center);
		\draw (15.center) to (18.center);
		\draw (25.center) to (24);
		\draw (33.center) to (32);
		\draw [in=75, out=-180] (34.center) to (32);
		\draw [in=-75, out=180] (35.center) to (32);
		\draw (36.center) to (35.center);
		\draw (34.center) to (37.center);
		\draw (43.center) to (42);
	\end{pgfonlayer}
\end{tikzpicture} 
\] 

More generally, if $p$ and $f$ depend on an additional parameter $A$, we can form the \emph{joint morphism} $fp:A\to X\otimes Y$ as follows.
\begin{equation}\label{joint}
 \begin{tikzpicture}
	\begin{pgfonlayer}{nodelayer}
		\node [style=medium box] (0) at (-1, 1) {$p$};
		\node [style=medium box] (1) at (2.25, -0.5) {$f$};
		\node [style=bn] (2) at (0.5, 1) {};
		\node [style=none] (3) at (-1.5, 1) {};
		\node [style=none] (4) at (-4, 0) {$A$};
		\node [style=none] (5) at (1.5, 2) {};
		\node [style=none] (6) at (1.5, 0) {};
		\node [style=none] (7) at (3.5, -0.5) {};
		\node [style=none] (8) at (3.5, 2) {};
		\node [style=none] (9) at (4, 2) {$X$};
		\node [style=none] (10) at (4, -0.5) {$Y$};
		\node [style=bn] (11) at (-2.5, 0) {};
		\node [style=none] (12) at (-1.5, -1) {};
		\node [style=none] (13) at (-3.5, 0) {};
		\node [style=none] (14) at (1.5, -1) {};
	\end{pgfonlayer}
	\begin{pgfonlayer}{edgelayer}
		\draw (3.center) to (2);
		\draw [in=75, out=-180] (5.center) to (2);
		\draw [in=-75, out=180] (6.center) to (2);
		\draw (5.center) to (8.center);
		\draw [in=180, out=75] (11) to (3.center);
		\draw [in=-75, out=180] (12.center) to (11);
		\draw (13.center) to (11);
		\draw (14.center) to (12.center);
		\draw (7.center) to (1);
	\end{pgfonlayer}
\end{tikzpicture}
\end{equation}
(In formulas, $fp=(\id_X\otimes f)\circ(\cop_X\otimes\id_A)\otimes(p\otimes\id_A)\otimes\cop_A$.)
This corresponds, for finite probability measures, to the formula
\[
 fp\,(x,y|a) \coloneqq p(x|a) \, f(y|x,a) .
\]
As in this equation both $x$ and $a$ appear twice, in \eqref{joint} the copy maps of both $X$ and $A$ are used.
Once again, it is easy to check that the marginals of this joint morphism are $(fp)_X=p$ and $(fp)_Y=f\circ p$.

We are now ready for the equivalent characterization of Markov categories with divergences.

\begin{theorem}\label{equivenrich}
 Let $\cat{C}$ be a Markov category equipped with a divergence $D$ on each hom-set $\cat{C}(X,Y)$. 
 Then the conditions of \Cref{defdivmarkov}, together, are equivalent to the following conditions, together:
 \begin{enumerate}
  \item\label{dataprocmarg} For any $f,f':X\to Y\otimes Z$, 
  we have that if we take the marginals on $Z$,
  \[
   D(f_Z\et f'_Z) \le D(f\et f') ;
  \]
  \item\label{pastejoints} Given $f,f':X\to Y$ and $g,g':X\otimes Y\to Z$,
  the following inequality holds for the joint morphisms,
  \[
   D(gf\et g'f') \le D(f\et f') + D(g\et g') . 
  \]
  \item\label{prodmarg} Given $f,f':X\to Y$ and any object $A$, we have that 
  \[
   D(f\otimes\discard_A\et f'\otimes\discard_A) \le D(f\et f') .
  \]
 \end{enumerate}
\end{theorem}

Here is how to interpret the third condition of \Cref{equivenrich}. Recall that, by using the discard maps, we can treat $f,f':X\to Y$ equivalently as channels $X\otimes A\to Y$ which don't really depend on $A$. The condition says that the divergence $D(f\et f')$ does not increase on whether we consider $f,f'$ as channels $X\otimes A\to Y$ instead of $X\to Y$. In particular, the divergence between constant functions cannot be more than the divergence between their constant values.
(If our divergences are in the form \eqref{Dtocat}, this is automatically true, more on that later.)

\begin{proof}[Proof of \Cref{equivenrich}.]
 First, suppose that the conditions \ref{dataprocmarg}--\ref{prodmarg} are satisfied. 
 To prove \eqref{enrichmentcondition}, consider $f,f':X\to Y$ and $g,g':Y\to Z$. Form the joint channel $\tilde{g}f$ as follows,
 
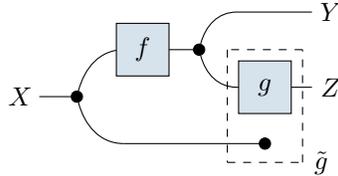
\[
  \begin{tikzpicture}
	\begin{pgfonlayer}{nodelayer}
		\node [style=medium box] (0) at (-0.75, 1.25) {$f$};
		\node [style=medium box] (1) at (2.5, 0.25) {$g$};
		\node [style=bn] (2) at (0.75, 1.25) {};
		\node [style=none] (3) at (-1.25, 1.25) {};
		\node [style=none] (4) at (-4, 0) {$X$};
		\node [style=none] (5) at (1.75, 2.25) {};
		\node [style=none] (6) at (1.75, 0.25) {};
		\node [style=none] (7) at (3.75, 0.25) {};
		\node [style=none] (8) at (3.75, 2.25) {};
		\node [style=none] (9) at (4.25, 2.25) {$Y$};
		\node [style=none] (10) at (4.25, 0.25) {$Z$};
		\node [style=bn] (11) at (-2.5, 0) {};
		\node [style=none] (12) at (-1.25, -1.25) {};
		\node [style=none] (13) at (-3.5, 0) {};
		\node [style=none] (14) at (2.5, -1.25) {};
		\node [style=bn] (15) at (2.5, -1.25) {};
		\node [style=none] (16) at (1.5, -1.75) {};
		\node [style=none] (17) at (3.5, -1.75) {};
		\node [style=none] (18) at (1.5, 1.25) {};
		\node [style=none] (19) at (3.5, 1.25) {};
		\node [style=none] (20) at (4, -1.75) {$\tilde{g}$};
	\end{pgfonlayer}
	\begin{pgfonlayer}{edgelayer}
		\draw (3.center) to (2);
		\draw [in=75, out=-180] (5.center) to (2);
		\draw [in=-75, out=180] (6.center) to (2);
		\draw (5.center) to (8.center);
		\draw [in=180, out=75] (11) to (3.center);
		\draw [in=-75, out=180] (12.center) to (11);
		\draw (13.center) to (11);
		\draw (14.center) to (12.center);
		\draw (6.center) to (7.center);
		\draw [style=dashed box] (19.center) to (18.center);
		\draw [style=dashed box] (18.center) to (16.center);
		\draw [style=dashed box] (16.center) to (17.center);
		\draw [style=dashed box] (17.center) to (19.center);
	\end{pgfonlayer}
\end{tikzpicture}
 \]
 in formulas $\tilde{g}=g\otimes\discard_X$,
 and form analogously $\tilde{g}'f'$. 
 Notice that the marginal $(\tilde{g}f)_Z$ is exactly $g\circ f$, and the same is true for the primed letters.
 Now we can use conditions \ref{dataprocmarg}, \ref{pastejoints}, and \ref{prodmarg}, in order, which gives us that 
 \begin{align*}
  D(g\circ f\et g\circ f') &= D\big((\tilde{g}f)_Z\et (\tilde{g}'f')_Z\big)\\
  &\le D(\tilde{g}f\et \tilde{g}'f') \\
  &\le D(f\et f') + D(\tilde{g}\et \tilde{g}') \\
  &= D(f\et f') + D(g\otimes\discard_X\et g'\otimes\discard_X) \\
  &\le D(f\et f') + D(g\et g') ,
 \end{align*}
 i.e.~condition~\eqref{enrichmentcondition}.
 
 To prove \eqref{tensortriangle}, consider channels $f,f':X\to Y$ and $h,h':A\to B$. We can now treat the product $f\otimes h$ as a particular joint morphism of $\tilde{f}:X\otimes A\to Y$ and $\tilde{g}:(X\otimes A)\otimes Y\to Z$ as follows.
 
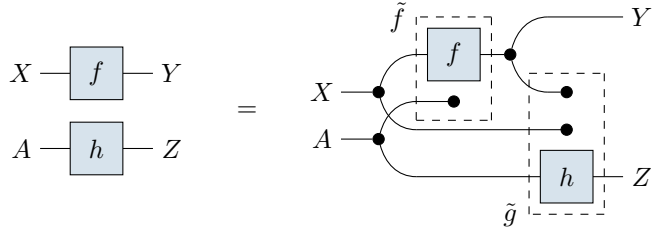
\[ 
  \begin{tikzpicture}
	\begin{pgfonlayer}{nodelayer}
		\node [style=bn] (2) at (7, 1.5) {};
		\node [style=none] (3) at (4.5, 1.5) {};
		\node [style=none] (4) at (2, 0.5) {$X$};
		\node [style=none] (5) at (8, 2.5) {};
		\node [style=none] (6) at (8, 0.5) {};
		\node [style=none] (7) at (10, -1.75) {};
		\node [style=none] (8) at (10, 2.5) {};
		\node [style=none] (9) at (10.5, 2.5) {$Y$};
		\node [style=none] (10) at (10.5, -1.75) {$Z$};
		\node [style=bn] (11) at (3.5, 0.5) {};
		\node [style=none] (12) at (4.5, -0.5) {};
		\node [style=none] (13) at (2.5, 0.5) {};
		\node [style=none] (14) at (8.5, -0.5) {};
		\node [style=none] (15) at (2, -0.75) {$A$};
		\node [style=none] (16) at (2.5, -0.75) {};
		\node [style=bn] (17) at (3.5, -0.75) {};
		\node [style=none] (18) at (4.5, 0.25) {};
		\node [style=none] (19) at (4.5, -1.75) {};
		\node [style=none] (20) at (5.5, 0.25) {};
		\node [style=bn] (21) at (8.5, -0.5) {};
		\node [style=bn] (22) at (5.5, 0.25) {};
		\node [style=bn] (23) at (8.5, 0.5) {};
		\node [style=none] (24) at (4.5, 2.5) {};
		\node [style=none] (25) at (4.5, -0.25) {};
		\node [style=none] (26) at (6.5, -0.25) {};
		\node [style=none] (27) at (6.5, 2.5) {};
		\node [style=none] (28) at (7.5, 1) {};
		\node [style=none] (29) at (9.5, 1) {};
		\node [style=none] (30) at (7.5, -2.75) {};
		\node [style=none] (31) at (9.5, -2.75) {};
		\node [style=none] (34) at (4, 2.5) {$\tilde{f}$};
		\node [style=none] (35) at (7, -2.75) {$\tilde{g}$};
		\node [style=none] (36) at (0, 0) {$=$};
		\node [style=none] (39) at (-6, 1) {$X$};
		\node [style=none] (40) at (-5.5, 1) {};
		\node [style=none] (42) at (-2.5, -1) {};
		\node [style=none] (43) at (-2.5, 1) {};
		\node [style=none] (44) at (-2, 1) {$Y$};
		\node [style=none] (45) at (-2, -1) {$Z$};
		\node [style=none] (50) at (-6, -1) {$A$};
		\node [style=none] (54) at (-5.5, -1) {};
		\node [style=medium box] (69) at (-4, 1) {$f$};
		\node [style=medium box] (70) at (-4, -1) {$h$};
		\node [style=medium box] (71) at (5.5, 1.5) {$f$};
		\node [style=medium box] (72) at (8.5, -1.75) {$h$};
	\end{pgfonlayer}
	\begin{pgfonlayer}{edgelayer}
		\draw (3.center) to (2);
		\draw [in=75, out=-180] (5.center) to (2);
		\draw [in=-75, out=180] (6.center) to (2);
		\draw (5.center) to (8.center);
		\draw [in=180, out=75] (11) to (3.center);
		\draw [in=-75, out=180] (12.center) to (11);
		\draw (13.center) to (11);
		\draw (14.center) to (12.center);
		\draw (17) to (16.center);
		\draw [in=75, out=-180] (18.center) to (17);
		\draw [in=180, out=-75] (17) to (19.center);
		\draw (20.center) to (18.center);
		\draw [style=dashed box] (24.center) to (25.center);
		\draw [style=dashed box] (25.center) to (26.center);
		\draw [style=dashed box] (26.center) to (27.center);
		\draw [style=dashed box] (27.center) to (24.center);
		\draw (6.center) to (23);
		\draw [style=dashed box] (28.center) to (29.center);
		\draw [style=dashed box] (29.center) to (31.center);
		\draw [style=dashed box] (31.center) to (30.center);
		\draw [style=dashed box] (30.center) to (28.center);
		\draw (7.center) to (19.center);
		\draw (40.center) to (43.center);
		\draw (42.center) to (54.center);
	\end{pgfonlayer}
\end{tikzpicture} 
 \] 
 (In formulas, $\tilde{f}\coloneqq f\otimes\discard_A$ and $\tilde{g}\coloneqq\discard_Y\otimes\discard_X\otimes h$.) 
 Form $\tilde{f}'$ and $\tilde{g}'$ analogously. Now conditions \ref{pastejoints} and \ref{prodmarg} imply that
 \begin{align*}
  D(f\otimes h\et f'\otimes h') &= D(\tilde{g}\tilde{f}\et \tilde{g}'\tilde{f}') \\
  &\le D(\tilde{f}\et \tilde{f}') + D(\tilde{g}\et \tilde{g}') \\
  &\le D(f\et f') + D(h\et h') ,  
 \end{align*}
 i.e.~condition~\eqref{tensortriangle}.
 
 Conversely, suppose that the conditions \eqref{enrichmentcondition} and \eqref{tensortriangle} of \Cref{defdivmarkov} are satisfied. Then \ref{dataprocmarg} follows from \eqref{enrichmentcondition} by taking as $g,g'$ the marginalization on $Z$,
 \[
  \begin{tikzpicture}
	\begin{pgfonlayer}{nodelayer}
		\node [style=none] (10) at (-0.75, 0.5) {};
		\node [style=none] (11) at (-0.75, -0.5) {};
		\node [style=none] (12) at (0.75, -0.5) {};
		\node [style=none] (13) at (0.5, 0.5) {};
		\node [style=none] (14) at (-1.25, 0.5) {$Y$};
		\node [style=none] (15) at (1.25, -0.5) {$Z$};
		\node [style=none] (30) at (-1.25, -0.5) {$Z$};
		\node [style=bn] (31) at (0.5, 0.5) {};
	\end{pgfonlayer}
	\begin{pgfonlayer}{edgelayer}
		\draw (10.center) to (13.center);
		\draw (11.center) to (12.center);
	\end{pgfonlayer}
\end{tikzpicture}
 \]
 in formulas, $\discard_X\otimes\id_Z$.
 Similarly, \ref{prodmarg} follows from \eqref{tensortriangle} by taking as $h,h'$ the marginalization on $Y$.
 To prove condition~\ref{pastejoints}, let $f,f':X\to Y$ and $g,g':X\otimes Y\to Z$. 
 We can write the joint morphism $gf$ as the following sequential composition.
 
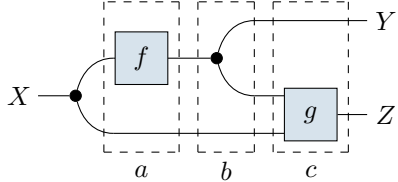
\[ 
  \begin{tikzpicture}
	\begin{pgfonlayer}{nodelayer}
		\node [style=medium box] (0) at (-1.5, 1) {$f$};
		\node [style=medium box] (1) at (3, -0.5) {$g$};
		\node [style=bn] (2) at (0.5, 1) {};
		\node [style=none] (3) at (-2.25, 1) {};
		\node [style=none] (4) at (-4.75, 0) {$X$};
		\node [style=none] (5) at (1.5, 2) {};
		\node [style=none] (6) at (1.5, 0) {};
		\node [style=none] (7) at (4.5, -0.5) {};
		\node [style=none] (8) at (4.5, 2) {};
		\node [style=none] (9) at (5, 2) {$Y$};
		\node [style=none] (10) at (5, -0.5) {$Z$};
		\node [style=bn] (11) at (-3.25, 0) {};
		\node [style=none] (12) at (-2.25, -1) {};
		\node [style=none] (13) at (-4.25, 0) {};
		\node [style=none] (14) at (2.5, -1) {};
		\node [style=none] (15) at (-2.5, 2.5) {};
		\node [style=none] (16) at (-0.5, 2.5) {};
		\node [style=none] (17) at (-2.5, -1.5) {};
		\node [style=none] (18) at (-0.5, -1.5) {};
		\node [style=none] (19) at (1.5, 2.5) {};
		\node [style=none] (20) at (0, 2.5) {};
		\node [style=none] (21) at (0, -1.5) {};
		\node [style=none] (22) at (1.5, -1.5) {};
		\node [style=none] (23) at (2, 2.5) {};
		\node [style=none] (24) at (2, -1.5) {};
		\node [style=none] (25) at (4, 2.5) {};
		\node [style=none] (26) at (4, -1.5) {};
		\node [style=none] (27) at (2.5, 0) {};
		\node [style=none] (28) at (-1.5, -2) {$a$};
		\node [style=none] (29) at (0.75, -2) {$b$};
		\node [style=none] (30) at (3, -2) {$c$};
	\end{pgfonlayer}
	\begin{pgfonlayer}{edgelayer}
		\draw (3.center) to (2);
		\draw [in=75, out=-180] (5.center) to (2);
		\draw [in=-75, out=180] (6.center) to (2);
		\draw (5.center) to (8.center);
		\draw [in=180, out=75] (11) to (3.center);
		\draw [in=-75, out=180] (12.center) to (11);
		\draw (13.center) to (11);
		\draw (14.center) to (12.center);
		\draw (7.center) to (1);
		\draw [style=dashed box] (15.center) to (17.center);
		\draw [style=dashed box] (17.center) to (18.center);
		\draw [style=dashed box] (18.center) to (16.center);
		\draw [style=dashed box] (16.center) to (15.center);
		\draw [style=dashed box] (20.center) to (21.center);
		\draw [style=dashed box] (21.center) to (22.center);
		\draw [style=dashed box] (22.center) to (19.center);
		\draw [style=dashed box] (19.center) to (20.center);
		\draw [style=dashed box] (26.center) to (24.center);
		\draw [style=dashed box] (24.center) to (23.center);
		\draw [style=dashed box] (23.center) to (25.center);
		\draw [style=dashed box] (25.center) to (26.center);
		\draw (6.center) to (27.center);
	\end{pgfonlayer}
\end{tikzpicture} 
 \]
 (In formulas, $a=f\otimes\id_X$, $b=\cop_X\otimes\id_X$, $c=\id_Y\otimes g$.)
 We can do the same for $f'$ and $g'$.
 Iterating \eqref{enrichmentcondition}, and using \eqref{tensortriangle}, we have that
 \begin{align*}
  D(gf\et g'f') &= D(c\circ b\circ a \circ \cop_X\et  c'\circ b'\circ a' \circ \cop_X) \\
  &\le D(c\et c') + D(b\et b') + D(a\et a') + 0 \\
  &= \begin{multlined}[t]
      D(f\otimes\id_X\et  f'\otimes\id_X) \\
       + D(\cop_X\otimes\id_X\et  \cop_X\otimes\id_X) \\
        + D(\id_Y\otimes g\et  \id_Y\otimes g')
     \end{multlined} \\
  &\le D(f\et f') + 0 + D(g\et g') ,
 \end{align*}
 which is exactly condition \ref{pastejoints}.
\end{proof}

\begin{corollary}\label{easychar}
 If $\cat{C}$ is $\cat{Stoch}$ or $\cat{FinStoch}$, and the divergences are given by taking the supremum over the inputs, as in \eqref{Dtocat}, conditions \ref{dataprocmarg} and \ref{pastejoints} are reduced to the following simpler form.
 \begin{enumerate}
  \item\label{dataprocmarg2} For any (joint) sources $p,p'$ on $X\otimes Y$ forming the marginals on $X$,
  \[
   D(p_X\et p'_X) \le D(p\et p') ;
  \]
  \item\label{pastejoints2} Given sources $p,p'$ on $X$ and channels $f,f':X\to Y$,
  the following inequality holds for the joint sources,
  \[
   D(fp\et f'p') \le D(p\et p') + \sup_{x\in X} D(f_x\et f'_x) . 
  \]
 \end{enumerate}
 In this context, condition \ref{prodmarg} of \Cref{equivenrich} is automatically satisfied.
\end{corollary}

We can interpret the first condition as a data processing inequality, as well as a \emph{monotonicity} condition for the divergence $D$, in the sense that additional data give additional distinctions. In terms of random variables, it reads 
\[
 D(X\et X') \le D(X,Y\et X',Y') .
\]
The second condition, instead, can be interpreted as a \emph{generalized chain rule} for the divergence $D$, which in terms of random variables would read as follows,
\[
 D(X,Y\et X',Y') \le D(X\et X') + \sup_{x\in X} D(Y\et Y'|X=x) .
\]

\subsection{Particular divergences}

Several of the divergences used in information theory, probability, and statistics, are examples of divergences on $\cat{Stoch}$ and $\cat{FinStoch}$ in the sense of \Cref{defdivmarkov}.
Here are some examples, and a nonexample (\Cref{nonexample}). A complete classification of all the divergences on $\cat{Stoch}$ and $\cat{FinStoch}$ is for now still an open question.

\subsubsection{The KL divergence (relative entropy)}

Recall the relative entropy from \Cref{defKL}. 

\begin{definition}
 The \emph{relative entropy} or \emph{KL} divergence on $\cat{FinStoch}$ is defined as follows for each pair of morphisms $f,g:X\to Y$,
 \[
 D(f\et g) \coloneqq \max_{x\in X} D_{KL}\big( f(x)\et g(x) \big) = \max_{x\in X} \sum_{y\in Y} f(y|x) \ln \dfrac{f(y|x)}{g(y|x)},
 \]
 again using \Cref{logzero}.
\end{definition}

As it is well known, this quantity is always positive \cite[Theorem~2.6.3]{cover-thomas}. 

\begin{proposition}\label{enrichedKL}
 The relative entropy is a divergence on $\cat{FinStoch}$. 
\end{proposition}

We prove this proposition using the \emph{chain rule for relative entropy} \cite[Theorem~2.5.3]{cover-thomas}: given joint distributions $p$ and $p'$ on $X\times Y$, we have that 
 \begin{equation}\label{chainruleKL}
  D(p\et p') = D(p_X\et p'_X) + \sum_{x\in X} p_X(x)\,D(p_{Y|x}\et p'_{Y|x}) ,
 \end{equation}
 where $p_X$ and $p'_X$ denote the marginal distributions on $X$, and $p_{Y|x}$ and $p'_{Y|x}$ denote the conditional distributions on $Y$ depending on $x$.
 The second term on the right is sometimes called \emph{conditional relative entropy} (see again \cite[Section~2.5]{cover-thomas})

\begin{proof}[Proof of \Cref{enrichedKL}]
 We can use the convenient characterization of \Cref{easychar}.  
 To prove condition~\ref{dataprocmarg2}, the chain rule~\eqref{chainruleKL} implies immediately that 
 \[
  D(p_X\et p'_X) \le D(p\et p') .
 \]
 
 To prove condition~\ref{pastejoints2}, given $p,p'\in PX$ and channels $f,f':X\to Y$, apply the chain rule \eqref{chainruleKL} for the joints $fp$ and $f'p'$, obtaining
 \[
  D(fp\et f'p') = D(p\et p') + \sum_{x\in X} p(x)\,D(f_x\et f'_x) .
 \]
 The last term is a convex combination indexed by $x$, and is hence bounded by its largest term,
 \[
  D(fp\et f'p') \le D(p\et p') + \sup_{x\in X} D(f_x\et f'_x) . \qedhere
 \] 
\end{proof}

\begin{remark}
The inequality \eqref{simplertens}, for relative entropy, is an equality: given finite probability measures $p$ and $p'$ on $X$ and $q$ and $q'$ on $A$, we have that 
\[
 D(p\otimes q\et p'\otimes q') = D(p\et p') + D(q\et q') ,
\]
as an easy calculation shows.
\end{remark}

Let's now extend these results to the infinite case ($\cat{Stoch}$). 
Let $X$ be a measurable space with $\sigma$-algebra $\Sigma_X$. Given probability measures $p$ and $q$ on $X$, we say that $p$ is \emph{absolutely continuous} with respect to $q$, and we write $p\ll q$, if whenever for a measurable set $S\in \Sigma_X$, $q(S)=0$, then also $p(S)=0$.
The Radon-Nikodym theorem~\cite[Section~3.2]{bogachev} says that if (and only if) $p\ll q$, we can find a measurable function $g:X\to\R$ such that for every $S\in\Sigma_X$,
\[
 p(S) = \int_S g\, dq .
\]
This function $g$, which is uniquely defined $q$-almost everywhere, is called the \emph{Radon-Nikodym derivative} of $p$ w.r.t.~$q$ and is denoted by $dp/dq$. 
Using this, relative entropy can be extended to general measurable spaces as follows.

\begin{definition}
 Let $p$ and $q$ be probability measures on a measurable space $X$. The \emph{relative entropy} or \emph{Kullback-Leibler divergence} between $p$ and $q$ is given by 
 \[
  D(p\et q) \coloneqq \int_X \ln \left( \dfrac{dp}{dq} \right)\,dp ,
 \]
 if $p\ll q$, and $D(p\et q)=\infty$ otherwise.
\end{definition}

We can as usual suprematize over the inputs, and obtain a divergence between Markov kernels. 
\begin{equation}\label{KLStoch}
 D(f\et g) \coloneqq \sup_{x\in X} D(f_x,g_x) =  \sup_{x\in X} \int_Y \ln \left( \dfrac{df_x}{dg_x}(x,y)\right) f(dy|x) ,
\end{equation}
if $f_x\ll g_x$ for all $x\in X$, and $\infty$ otherwise.

\begin{proposition}\label{KLStochdiv}
 The divergence in \eqref{KLStoch} is a divergence on $\cat{Stoch}$. 
\end{proposition}

In order to prove the proposition, we first need a couple of technical lemmas. 

\begin{lemma}\label{jointdens}
 Let $p$ and $p'$ be measures on a measurable space $X$, and let $f$ and $f'$ be kernels $X\to Y$. Suppose that $p\ll p'$ and that for $p'$-almost all $x\in X$, $f_x\ll f'_x$. Then $fp\ll f'p'$, and for $f'p'$-almost all $x\in X$ and $y\in Y$,
 \begin{equation}\label{jdeq}
  \dfrac{d(fp)}{d(f'p')}(x,y) = \dfrac{dp}{dp'}(x)\,\dfrac{df_x}{df'_x}(x,y) .
 \end{equation}
\end{lemma}

\begin{proof}[Proof of \Cref{jointdens}.]
 First of all, for each measurable subset $S$ of $X\times Y$, 
 \[
  fp\,(S) = \int_X f_x(S_x) \, p(dx) ,
 \]
 where for each $x\in X$, we denote by $S_x$ the set $\{y\in Y : (x,y)\in S\}$, which is a measurable subset of $Y$.
 Similarly, 
 \[
  f'p'\,(S) = \int_X f'_x(S_x) \, p'(dx) .
 \]
 Suppose now that $f'p'\,(S)=0$. Then in the equation above, the set $T=\{x\in~X:f'_x(S_x) \ne 0\}$ must have $p'$-measure zero, and since $p\ll p'$, this set has also $p$-measure zero. 
 Moreover, for all $x\in X\setminus T$ we have $f'_x(S_x)=0$, and since $f_x\ll f'_x$ for $p'$-almost all $x$, we also must have $f_x(S_x)=0$ for all $x$ in $X\setminus T$ up to a $p'$-null set (hence, $p$-null set). Therefore 
 \[
  fp\,(S) = \int_X f_x(S_x) \, p(dx) = \int_{X\setminus T} f_x(S_x) \, p(dx) = 0 ,
 \]
 which means $fp\ll f'p'$. 
 
 Now, every Radon-Nikodym derivative of $fp$ w.r.t.~$f'p'$ must satisfy
 \[
  fp\,(S) = \int_S \dfrac{d(fp)}{d(f'p')} \,d(fp) 
 \]
 for each measurable subset $S$ of $X\times Y$. 
 We have that 
 \begin{align*}
  \int_S \dfrac{dp}{dp'}(x)\,\dfrac{df_x}{df'_x}(x,y) \, f'p'(dx\,dy) &= \int_X \int_{S_x} \dfrac{dp}{dp'}(x)\,\dfrac{df_x}{df'_x}(x,y) \, f'(dy|x)\,p'(dx) \\
  &= \int_X \dfrac{dp}{dp'}(x) \int_{S_x} \dfrac{df_x}{df'_x}(x,y) \, f'(dy|x)\,p'(dx) \\
  &= \int_X \dfrac{dp}{dp'}(x) \, f_x(S_x)\,p'(dx) \\
  &= \int_X f_x(S_x)\,p(dx) \\
  &= pf\,(S) .
 \end{align*}
 Therefore, by the Radon-Nikodym theorem, the two sides of \eqref{jdeq} are equal $f'p'$-almost everywhere.
\end{proof}

\begin{lemma}\label{KLcontcr}
 Let $p$ and $p'$ be measures on a measurable space $X$, and let $f$ and $f'$ be kernels $X\to Y$. Then we have the following \emph{chain rule}:
 \begin{equation}\label{KLcontcreq}
  D(fp\et f'p') = D(p\et p') + \int_X D(f_x\et f'_x)\, p(dx)
 \end{equation}
\end{lemma}

One can call the last term, analogously to the discrete case, the \emph{conditional relative entropy}. 

\begin{proof}[Proof of \Cref{KLcontcr}]
 We can assume $p\ll p'$ and $f_x\ll f'_x$ for all $x$, otherwise condition~\ref{pastejoints2} holds immediately.
 By \Cref{jointdens},
 \begin{align*}
  D(fp\et f'p') &= \int_{X\times Y} \ln\left( \dfrac{d(fp)}{d(f'p')}(x,y) \right)  fp\,(dx\,dy) \\
   &= \int_X \int_Y \ln\left( \dfrac{dp}{dp'}(x)\,\dfrac{df_x}{df'_x}(x,y) \right)  f(dy|x)\, p(dx) \\
   &= \int_X \int_Y \left( \ln\dfrac{dp}{dp'}(x) + \ln\dfrac{df_x}{df'_x}(x,y) \right)  f(dy|x)\, p(dx) \\
   &= \int_X  \ln \left( \dfrac{dp}{dp'}  \right)\,dp + \int_X \int_Y \left( \ln\dfrac{df_x}{df'_x}(x,y) \right) f(dy|x)\, p(dx)\\
   &= D(p\et p') + \int_X D(f_x\et f'_x)\, p(dx) . \qedhere
 \end{align*}
\end{proof}

We can now prove the proposition.

\begin{proof}[Proof of \Cref{KLStochdiv}.]
 As for the discrete case, we can use the characterization of \Cref{easychar}.
 Condition~\ref{dataprocmarg2} holds since a more general data processing inequality holds, for general measurable functions, not just marginalizations~\cite[Theorem~9]{renyi-div}.
 
 To prove condition~\ref{pastejoints2}, we can use \Cref{KLcontcr}. Let $p,p'\in PX$ and kernels $f,f':X\to Y$. Then in \eqref{KLcontcreq}, once again, we can bound the last term by taking the supremum over $x$, so that 
 \[
  D(fp\et f'p') \le D(p\et p') + \sup_x D(f_x\et f'_x) . \qedhere
 \]
\end{proof}

\subsubsection{The Rényi or alpha-divergence enrichments.}

The Rényi divergence is a generalization or deformation of the Kullback-Leibler divergence, and it is related to the Rényi entropy. 
A comprehensive analysis of this divergence, which includes all the results used in this section, can be found in \cite{renyi-div}. 

\begin{definition}
 Let $X$ be a finite set, let $p$ and $q$ be probability distributions on $X$, and let $\alpha\in [0,\infty]$. 
 We define the \emph{Rényi divergence of order $\alpha$} or \emph{$\alpha$-divergence} between $p$ and $q$, $D_\alpha(p\et q)$, as follows. 
 
 First of all, if $p(x)=0$ but $q(x) \ne 0$ for some $x\in X$, we set $D_{\alpha}(p\et q) \coloneqq \infty$ for all $\alpha$. 
 Instead, if $p(x)=0$ whenever $q(x)=0$,
 \begin{itemize}
  \item For $\alpha> 0, \alpha\ne 1$, $D_{\alpha}(p\et q)$ is the quantity
 \begin{equation}
  D_{\alpha}(p\et q) \coloneqq \dfrac{1}{\alpha-1} \, \ln \left( \sum_{x\in X} \dfrac{p(x)^\alpha}{q(x)^{\alpha-1}} \right) ;
 \end{equation}
 \item For $\alpha=1$, we set $D_1(p\et q) \coloneqq D_{KL}(p\et q)$, the relative entropy; 
 \item For $\alpha=0$, we set $D_0(p\et q) \coloneqq \lim_{\alpha\to 0} D_\alpha(p\et q)$;
 \item For $\alpha=+\infty$, we set $D_\infty(p\et q) \coloneqq \lim_{\alpha\to\infty} D_\alpha(p\et q)$.
 \end{itemize}
\end{definition}

Note that indeed, 
\[
\lim_{\alpha\to 1} D_\alpha(p\et q) = D_{KL}(p\et q) ,
\]
see \cite[Section~II.C]{renyi-div} for more on this.

Just as for relative entropy, one can generalize this definition to measurable spaces as follows.
\begin{definition}
 Let $X$ be a measurable space, and let $p$ and $q$ be probability measures on $X$. Given $\alpha> 0, \alpha\ne 1$, the \emph{Rényi divergence of order $\alpha$} or \emph{$\alpha$-divergence} between $p$ and $q$ is the number
 \begin{equation}
  D_{\alpha}(p\et q) \coloneqq \dfrac{1}{\alpha-1} \, \ln \left( \int_X \left( \dfrac{dp}{dq} \right)^{\alpha-1} dp \right) 
 \end{equation}
 if $p\ll q$, and $\infty$ otherwise.
 
 For $\alpha=1$, we set $D_1(p\et q) \coloneqq D_{KL}(p\et q)$, the relative entropy. 
 For $\alpha=0$, we set $D_0(p\et q) \coloneqq \lim_{\alpha\to 0} D_\alpha(p\et q)$.
 For $\alpha=+\infty$, we set $D_\infty(p\et q) \coloneqq \lim_{\alpha\to\infty} D_\alpha(p\et q)$.
\end{definition}
Once again, 
\[
\lim_{\alpha\to 1} D_\alpha(p\et q) = D_{KL}(p\et q) .
\]

\begin{proposition}\label{alphadiv}
 For all $\alpha\in[0,1]$, the Rényi divergence $D_\alpha$ induces a divergence on $\cat{Stoch}$.
\end{proposition}

In order to prove the proposition we make use of the following \emph{logarithmic chain rule} for the $\alpha$-divergence. 

\begin{lemma}\label{chainalpha}
 Let $p$ and $p'$ be probability measures on $PX$, and let $f,f':X\to Y$ be kernels.
 Suppose that $p\ll p'$ and that for $p'$-almost all $x\in X$, $f_x\ll f'_x$. 
 Then for $\alpha\in (0,\infty), \alpha\ne 1$,
 \begin{equation}\label{chainalphaeq}
  D_\alpha(fp\et f'p') = \dfrac{1}{\alpha-1} \, \ln \left( \int_X \left( \dfrac{dp}{dp'}(x) \right)^{\alpha-1} e^{(\alpha-1) D_\alpha(f_x\et f'_x)} \, p(dx) \right) .
 \end{equation} 
 (If one chooses the base of the logarithm to be anything other than $e$, one has to replace $e$ in the equation by the base of the logarithm.)
\end{lemma}

Equivalently, 
\[
 e^{(\alpha-1)\,D_\alpha(fp\et f'p')} = \int_X \left( \dfrac{dp}{dp'}(x) \right)^{\alpha-1} e^{(\alpha-1)\, D_\alpha(f_x\et f'_x)} \, p(dx) ,
\]
which means that the \emph{exponential} of the divergence between the joints is a weighted combination of the exponentials of the divergences between the channels.

\begin{proof}[Proof of \Cref{chainalpha}]
 Using \Cref{jointdens},
 \begin{align*}
  D_\alpha(fp\et f'p') &= \dfrac{1}{\alpha-1} \, \ln \left( \int_{X\times Y} \left( \dfrac{d(fp)}{d(f'p')}(x,y) \right)^{\alpha-1} fp\,(dx\,dy) \right) \\
  &= \dfrac{1}{\alpha-1} \, \ln \left( \int_X \int_Y \left( \dfrac{dp}{dp'}(x) \, \dfrac{df_x}{df'_x}(x,y) \right)^{\alpha-1} f(dy|x)\,p(dx) \right) \\
  &= \dfrac{1}{\alpha-1} \, \ln \left( \int_X \left( \dfrac{dp}{dp'}(x) \right)^{\alpha-1} \left( \int_Y \left(\dfrac{df_x}{df'_x} \right)^{\alpha-1} df_x \right) p(dx) \right) \\
  &= \dfrac{1}{\alpha-1} \, \ln \left( \int_X \left( \dfrac{dp}{dp'}(x) \right)^{\alpha-1} e^{(\alpha-1) D_\alpha(f_x\et f'_x)} \, p(dx) \right) . \qedhere
 \end{align*}
\end{proof}

\begin{proof}[Proof of \Cref{alphadiv}.]
 We can once again use the characterization of \Cref{easychar}.
 Just like for the relative entropy case, condition~\ref{dataprocmarg2} holds since a more general data processing inequality holds, for general measurable functions, not just marginalizations~\cite[Theorems~1 and~9]{renyi-div}.
 
 To prove condition~\ref{pastejoints2}, let $p$ and $p'$ be probability measure on $X$, and $f,f':X\to Y$ be kernels. 
 Just as for relative entropy, we can assume $p\ll p'$ and $f_x\ll f'_x$ for all $x$, otherwise condition~\ref{pastejoints2} holds immediately.
 
 Let first $\alpha>1$. By \Cref{chainalpha},
 \begin{align*}
  D_\alpha(fp\et f'p') &= \dfrac{1}{\alpha-1} \, \ln \left( \int_X \left( \dfrac{dp}{dp'}(x) \right)^{\alpha-1} e^{(\alpha-1) D_\alpha(f_x\et f'_x)} \, p(dx) \right) \\
  &\le \dfrac{1}{\alpha-1} \, \ln \left( \int_X \left( \dfrac{dp}{dp'}(x) \right)^{\alpha-1} \sup_{x'\in X} e^{(\alpha-1) D_\alpha(f_x\et f'_x)} \,p(dx) \right) \\
  &= \dfrac{1}{\alpha-1} \, \ln \left( \int_X \left( \dfrac{dp}{dp'}(x) \right)^{\alpha-1} p(dx) \cdot \sup_{x'\in X} e^{(\alpha-1) D_\alpha(f_x\et f'_x)} \right) \\
  &= \begin{multlined}[t]
      \dfrac{1}{\alpha-1} \, \ln \left( \int_X \left( \dfrac{dp}{dp'}(x) \right)^{\alpha-1} p(dx) \right) \\
       + \dfrac{1}{\alpha-1} \, \ln \left( \sup_{x\in X} e^{(\alpha-1) D_\alpha(f_x\et f'_x)} \right)
     \end{multlined} \\
  &= D_\alpha(p\et p') + \sup_{x\in X} D_\alpha(f_{x}\et f'_{x}) .
 \end{align*}
 
 For $\alpha<1$ the proof is similar, except that inside the logarithm one has to take the \emph{infimum} over $x$, which becomes a supremum after dividing by $\alpha-1$. 
 
 For $\alpha=0$ and $\alpha=\infty$, the argument follows by continuity.
 
 For $\alpha=1$ we have the KL divergence, for which this result has already been proven (\Cref{KLStochdiv}), or one can argue again by continuity. 
\end{proof}

\subsubsection{The total variation distance}

\begin{definition}
 Let $p$ and $q$ be probability distributions on a finite set $X$. The \emph{total variation distance} between $p$ and $q$ is given by 
 \[
  d_T(p,q) \coloneqq \dfrac{1}{2} \sum_{x\in X} |p(x) - q(x)| .
 \]
\end{definition}

Equivalently, in terms of measures of sets, we can write
\[
 d_T(p,q) = \sup_{A\subseteq X} |p(A) - q(A)| .
\]

Once again, for kernels $f,g:X\to Y$, we can define the distance as the maximum over the inputs,
 \[
  d_T(f,g) \coloneqq \dfrac{1}{2} \max_{x\in X} \sum_{y\in Y} \big| f(y|x) - g(y|x) \big| .
 \]

This distance generalizes to the infinite case as follows.

\begin{definition}
 Let $(X,\Sigma_X)$ be a measurable space, and let $p$ and $q$ be probability measures on $X$. 
 The \emph{total variation distance} between $p$ and $q$ is given by 
 \[
  d_T(p,q) \coloneqq \sup_{S\in\Sigma_X} \big| p(S) - q(S) \big| .
 \]
\end{definition}

An equivalent characterization of this distance is as follows,
\begin{equation}\label{totvarfunc}
 d_T(p,q) = \sup_{f:X\to [0,1]} \left| \int_X f\, dp - \int_X f\,dq \right| ,
\end{equation}
where the supremum is taken over all measurable functions $f:X\to [0,1]$. Notice that these functions include the kernels in the form $k(S|x)$, from $X$ to another space $Z$, for every fixed $S$, since such kernels are measurable in the variable $x$, and take values in $[0,1]$.

For kernels $f,g:X\to Y$, we can define the distance as the maximum over the inputs,
 \[
  d_T(f,g) \coloneqq \sup_{x\in X} \sup_{S\in\Sigma_Y} \big| f(S|x) - g(S|x) \big| .
 \]

\begin{proposition}
 The total variation distance gives a divergence on $\cat{Stoch}$.
\end{proposition}

\begin{proof}
 As we did for the KL and Rényi divergences, we use the characterization of \Cref{easychar}.
 First of all, given a measurable subset $S\subseteq X\times Y$, denote by $S_x$ the measurable subset of $Y$ given by 
 \[
    S_x \coloneqq \{y\in Y : (x,y) \in S \} ,
 \]
 and define $S_y\subseteq X$ analogously, so that for probability measures $p$ on $X$ and $q$ on $Y$, the product measures gives
 \[
  p\otimes q \, (S) = \int_X q(S_x)\,p(dx) = \int_Y p(S_y)\,q(dy) .
 \] 
 
 Now, in order to prove condition~\ref{dataprocmarg2}, let $p$ and $p'$ be probability measures on $X\times Y$. 
 We have 
 \begin{align*}
  d_T(p_X,p'_X) &=  \sup_{T\in\Sigma_{X}} \big| p_X(T) - p'_X(T) \big| \\
   &= \sup_{T\in\Sigma_{X}} \big| p(T\times Y) - p'(T\times Y) \big| \\
   &\le \sup_{S\in\Sigma_{X\times Y}}  \big| p(S) - p'(S) \big| \\
   &= d_T(p,p') .
 \end{align*}

 To prove condition~\ref{pastejoints2}, let $p$ and $p'$ be probability measure on $X$, and $f,f':~X\to Y$ be kernels.
 Recall that for measurable $S\subseteq X\times Y$,
 \[
    S_x \coloneqq \{y\in Y : (x,y) \in S \} ,
 \]
 and $S_y\subseteq X$ is defined analogously. 
 Then 
 \begin{align*}
  d_T(fp,f'p') &= \sup_{S\in\Sigma_{X\times Y}}  \big| pf\,(S) - p'f'\,(S) \big| \\
   &\le \sup_{S\in\Sigma_{X\times Y}}  \big| pf\,(S) - p'f\,(S) \big| + \sup_{S\in\Sigma_{X\times Y}} \big| p'f\,(S) - p'f'\,(S) \big| \\ 
   &= \begin{multlined}[t]
       \sup_{S\in\Sigma_{X\times Y}} \left| \int_X f(S_x|x) \, p(dx) - \int_X f(S_x|x) \, p'(dx) \right| \\
        + \sup_{S\in\Sigma_{X\times Y}} \left| \int_X f(S_x|x) \, p'(dx) - \int_X f'(S_x|x) \, p'(dx) \right| 
      \end{multlined} \\
   &\le \begin{multlined}[t]
       \sup_{g:X\to [0,1]} \left| \int_X g(x) \, p(dx) - \int_X g(x) \, p'(dx) \right| \\
        + \sup_{S\in\Sigma_{X\times Y}} \int_X \big| f(S_x|x) - f'(S_x|x) \big| p'(dx) 
      \end{multlined} \\
   &\le \begin{multlined}[t]
       \sup_{g:X\to [0,1]} \left| \int_X g(x) \, p(dx) - \int_X g(x) \, p'(dx) \right| \\
        + \int_X \sup_{T\in\Sigma_{Y}} \big| f(T|x) - f'(T|x) \big| \, p'(dx) 
      \end{multlined} \\
   &\le \begin{multlined}[t]
       \sup_{g:X\to [0,1]} \left| \int_X g(x) \, p(dx) - \int_X g(x) \, p'(dx) \right| \\
        + \sup_{x\in X} \sup_{T\in\Sigma_{Y}} \big| f(T|x) - f'(T|x) \big| 
      \end{multlined} \\
   &= d_T(p,p') + \sup_{x\in X} d_T(f_x,f'_x) . \qedhere
 \end{align*}
\end{proof}

\subsubsection{Nonexample: q-divergences}\label{nonexample}

Let's now see an example of divergence that does not give an enrichment.
Given probability measures $p$ and $p'$ on a finite set $X$, the \emph{Tsallis divergence} of order $q$, or more simply \emph{$q$-divergence}, is given by 
\[
 D_q(p\et p') \coloneqq - \sum_x p(x) \ln_q \dfrac{p'(x)}{p(x)} ,
\]
where the \emph{$q$-logarithm} is given by 
\[
 \ln_q x \coloneqq \dfrac{x^{1-q}-1}{1-q} 
\]
for $q\neq 1$,
and by the traditional natural logarithm for $q=1$.
Just as the Rényi divergence, for $q=1$ the Tsallis divergence is equal to the KL divergence. 

For $q\ne 1$, in general, the Tsallis divergence does not equip $\cat{FinStoch}$ with a divergence in the sense of \Cref{defdivmarkov}. 
For example, for $q=2$, consider the following distributions and kernels.
\[
 p = \begin{bmatrix}
      1/2      \\
      1/2
     \end{bmatrix}
     \qquad
 p' = \begin{bmatrix}
      3/4     \\
      1/4
     \end{bmatrix}
     \qquad
 f = \begin{bmatrix}
      1/4   & 9/10     \\
      3/4   & 1/10
     \end{bmatrix}
     \qquad
 f' = \begin{bmatrix}
      1/20  & 1/2    \\
      19/20 & 1/2
     \end{bmatrix}
\]
As one can readily check,
\[
 f\circ p = \begin{bmatrix}
      23/40    \\
      17/40
     \end{bmatrix}
     \qquad
 f'\circ p' = \begin{bmatrix}
      13/80    \\
      67/80
     \end{bmatrix}
\]
and for $q=2$ we have,
\[
 D_q(f\et f') = 1/3 
 \qquad
 D_q(g\et g') = 16/19 ,
\]
so
\[
 D_q(f\et f') + D_q(g\et g') = 67/57 \approx 1.175 ,
\]
but
\[
 D_q(f\circ p \et g\circ p) = 1089/871 \approx 1.250 .
\]

Similar counterexamples can more generally be found for other \emph{$f$-divergences} (see \cite[Chapter~4]{csizar} for the definitions). 
This is related to the well known fact that for Tsallis entropy and related quantities, additivity and subadditivity in general fail.

\subsection{Divergence as a limit over countable partitions}\label{partitions}

It is well known that divergences in the uncountable case can be obtained by taking the supremum over all countable partitions \cite[Theorems~2 and~10]{renyi-div}.
As we show here, this fact can be interpreted as an \emph{enriched} universal property.

\begin{definition}
 Let $X$ be a measurable space. A \emph{measurable partition} of $X$ is a family $S=\{S_i\}_{i\in I}$ of pairwise disjoint, measurable subsets $S_i\subseteq X$, such that $\coprod_i S_i = X$. The sets $S_i$ are called the \emph{cells} of the partition. 
 
 A partition is called \emph{countable} (resp.\ \emph{finite}) if it has countably (resp.\ finitely) many cells. 
\end{definition}

Note that in our definition cells are allowed to be empty (other authors may have other conventions).

\begin{definition}
 Given partitions $S=\{S_i\}_{i\in I}$ and $T=\{T_j\}_{j\in J}$ of $X$, we say that $S$ \emph{refines} $T$ if there exists a function $h:I\to J$ such that for all $j\in J$,
 \[
  T_j = \coprod_{i\in h^{-1}(j)} S_i . 
 \]
\end{definition}

Let's now focus our attention on the lattice of countable, measurable partitions of $X$. Denote their lattice by $\mathrm{CPart}(X)$.

\begin{definition}
 Let $X$ be a measurable space. A \emph{compatible family} of measures on $\mathrm{CPart}(X)$ amounts to 
 \begin{itemize}
  \item For each partition $S=\{S_i\}_{i\in I}$ on $X$, a discrete measure $\mu_S$ on $I$; such that
  \item Whenever $S=\{S_i\}_{i\in I}$ refines $T=\{T_j\}_{j\in J}$, via a function $h:I\to J$, we have that for all $j\in J$,
  \[
   \mu_T(T_j) = \sum_{i\in h^{-1}(j)} \mu_S(S_i) . 
  \]
 \end{itemize}

\end{definition}

\begin{lemma}\label{bijpart}
 Let $X$ be a measurable space. There is a bijective correspondence between
 \begin{itemize}
  \item probability measures $p$ on $X$, and
  \item compatible families $\{\mu_S\}$ of probability measures on $\mathrm{CPart}(X)$.
 \end{itemize}
 Moreover, the measure $p$ is zero-one if and only if every measure $p_S$ of the corresponding compatible family is zero-one. 
\end{lemma}

\begin{proof}[Proof of \Cref{bijpart}]
 First of all, every measure $p$ on $X$ defines a compatible family by restricting the measure to the $\sigma$-algebra generated by each partition, that is, for each cell $S_i$ of the partition $S$, and for each set of the induced $\sigma$-algebra, we set $p_S(S_i)=p(S_i)$.

 To show that this assignment is injective, suppose that the measures $p$ and $q$ on $X$ are different. Then there exists a measurable subset $A\subseteq X$ where $p(A)\ne q(A)$. Taking now the partition $S_A=\{A,X\setminus A\}$ we see that $p_S\ne q_S$, since they differ on $A$.
 
 To show that this assignment is surjective, let $\{\mu_S\}$ be a compatible family. Define the set function $p$ as follows. Given a measurable subset $A$ of $X$, consider the partition $S_A=\{A,X\setminus A\}$, and set $p(A)\coloneqq \mu_{S_A}(A)$. 
 We have that $p(\varnothing)=0$ since $S_\varnothing=\{\varnothing,X\}$. 
 To show countable additivity, let $\{A_i\}_{i=1}^\infty$ be pairwise disjoint measurable subsets of $X$, denote by $A$ their (disjoint) union.
 Then the partition $\tilde{S}=\{A_i\}\cup \{X\setminus A\}$ is a countable measurable partition of $X$ refining each of the partitions $S_{A_i}$ as well as $S_A$. Since the measure $\mu_S$ is by construction countably additive, and by compatibility of the family, we have that 
 \[
  p(A) = \mu_{S_A}(A) = \mu_{\tilde{S}}(A) = \sum_{i=1}^\infty \mu_{\tilde{S}}(A_i) = \sum_{i=1}^\infty \mu_{S_{A_i}}(A_i) = \sum_{i=1}^\infty p(A_i),
 \]
 and so $p$ is countably additive.
 
 To conclude the proof, we have that if $p$ is zero-one, so are all the pushforwards $p_S$, and conversely, if all the measures $p_S$ are zero one, then for each measurable $S\subseteq X$, $p(S)=p_S(S)$ has to be either zero or one as well.
\end{proof}

Note that given a partition $\{S_i\}_{i\in I}$ of $X$, we can consider the set $I$ equipped with the discrete $\sigma$-algebra.
We can also construct a measurable function $q:~X\to I$ which for all $i\in I$ maps each point in the cell $S_i\subseteq X$ to the point $i$.
Given a measure $p$ on $X$, the pushforward measure $q_*p$ has $q_*p(i) = p(S_i)$ for all $i\in I$. 
\Cref{bijpart} says that this makes $X$ a limiting cone over such partitions.

\begin{corollary}\label{limitcountable}
 $X$ is the limit of the diagram of its countable measurable partitions and their refinements, 
 both in the category $\cat{Stoch}$ and in its subcategory of zero-one morphisms.
\end{corollary}

This means that this limit is in particular a \emph{Markov limit}, similar to the Markov colimit of \cite[Section~3.2]{ergodic}.
(We will see in \Cref{H} that by \emph{deterministic morphisms} in $\cat{Stoch}$ we don't quite mean measurable functions, we mean more generally those kernels that only take values zero and one. Indeed, if one takes elements instead of measures, \Cref{bijpart} in general does not hold. See more in \cite{ergodic} as well as in \cite{ours_LICS}.) 

\begin{question}
 Is this limit preserved by the tensor product, analogously to the \emph{Kolmogorov products} of \emph{\cite{fritzrischel2019zeroone}}?
\end{question}

Not only is $X$ the limit of its countable partitions, but we can use those partitions to calculate divergences.

\begin{proposition}\label{Dsupcountable}
 For $D=D_{KL}$, $D=D_\alpha$ and $D=d_T$, and for probability measures $p$ and $q$ on a measurable space $X$,
 \[
  D(p\et q) = \sup_{S\in\mathrm{CPart}(X)} D(p_S \et q_S) .
 \]
 Moreover, the supremum can equivalently be taken over only the finite partitions.
\end{proposition}

For $D=D_{KL}$ and $D=D_\alpha$, this result is given by \cite[Theorems~2 and~10]{renyi-div}.
For $D=d_T$, it suffices to take binary partitions and suprematize over them. 

Therefore, one could have defined the divergence in the uncountable case, equivalently, as the supremum over the countable coarse-grainings. 

For the readers familiar with enriched categories, \Cref{limitcountable} and \Cref{Dsupcountable} imply together an enriched version of the limit property:
\begin{theorem}
 $X$ is the \emph{enriched} limit in $\cat{Stoch}$ of the diagram of its countable measurable partitions and their refinements, 
 where we consider $\cat{Stoch}$ enriched in divergences by means of \Cref{defdivmarkov} using either $D=D_{KL}$, $D=D_\alpha$ or $D=d_T$.
\end{theorem}

This shows that in some cases \emph{$\cat{Stoch}$ has enriched universal properties}. It is possible that other quantitative bounds in probability and information theory can be interpreted in this light too, and that new quantitative bounds can be founds by means of this framework.
For example, a quantitative extension of Kolmogorov's extension theorem would read as follows: \emph{the divergence between two probability distributions on an infinite cartesian product $X^\mathbb{N}$ is the supremum of the divergence between their corresponding finite marginalizations}. We leave such questions to future work.

\subsection{Almost-sure equality and conditional divergences}\label{conddiv}

The notion of almost-sure equality of channels can be given in any Markov category, see \cite[Section~13]{fritz2019synthetic} and \cite[Section~5]{chojacobs2019strings}.
Given channels $f,g:X\to Y$ and a source $p$ on $X$, in a Markov category, one says that $f$ and $g$ are $p$-almost surely equal if and only if the following equation holds.
\begin{equation}\label{aseq}
\begin{tikzpicture}
	\begin{pgfonlayer}{nodelayer}
		\node [style=none] (22) at (0, 0) {$=$};
		\node [style=horiz state] (36) at (-6.5, 0) {$p$};
		\node [style=bn] (37) at (-5.5, 0) {};
		\node [style=none] (38) at (-4.5, 1) {};
		\node [style=none] (39) at (-4.5, -1) {};
		\node [style=none] (40) at (-2.5, 1) {};
		\node [style=none] (41) at (-2.5, -1) {};
		\node [style=none] (42) at (-2, 1) {$X$};
		\node [style=none] (43) at (-2, -1) {$Y$};
		\node [style=medium box] (44) at (-3.75, -1) {$f$};
		\node [style=horiz state] (45) at (2.5, 0) {$p$};
		\node [style=bn] (46) at (3.5, 0) {};
		\node [style=none] (47) at (4.5, 1) {};
		\node [style=none] (48) at (4.5, -1) {};
		\node [style=none] (49) at (6.5, 1) {};
		\node [style=none] (50) at (6.5, -1) {};
		\node [style=none] (51) at (7, 1) {$X$};
		\node [style=none] (52) at (7, -1) {$Y$};
		\node [style=medium box] (53) at (5.25, -1) {$g$};
	\end{pgfonlayer}
	\begin{pgfonlayer}{edgelayer}
		\draw (38.center) to (40.center);
		\draw (41.center) to (39.center);
		\draw (36) to (37);
		\draw [in=180, out=75] (37) to (38.center);
		\draw [in=180, out=-75] (37) to (39.center);
		\draw (47.center) to (49.center);
		\draw (50.center) to (48.center);
		\draw (45) to (46);
		\draw [in=180, out=75] (46) to (47.center);
		\draw [in=180, out=-75] (46) to (48.center);
	\end{pgfonlayer}
\end{tikzpicture}
\end{equation}
Instantiating this in $\cat{FinStoch}$ we get that for all $x\in X$ and $y\in Y$, 
$$
p(x)\,f(y|x) = p(x)\,g(y|x) \,,
$$
i.e. $f(y|x)$ and $g(y|x)$ must agree on all the $x$ in the support of $p$. 
For $\cat{Stoch}$ and $\cat{BorelStoch}$ the situation is analogous, see the aforementioned references for the details.

In general, the divergence between both sides of \eqref{aseq}, i.e.\ $D(fp\et gp)$, can be seen as a measure of departure from the case of $f$ being $p$-almost surely equal to $g$. We can call the resulting quantity the \emph{conditional divergence}, and denote it by $D(f\et g\,|\,p)$. Let's see this for our usual examples.
\begin{itemize}
	\item For the KL divergence, we get from the chain rules \eqref{chainruleKL} and \eqref{KLcontcreq},
	$$
	D(f\et g\,|\,p) = D(fp\et gp) = \sum_{x\in X} p(x)\,D(f_x\et g_x)
	$$
	in the discrete case, and
	$$
	D(f\et g\,|\,p) = D(fp\et gp) = \int_X D(f_x\et g_x)\, p(dx)
	$$
	in the continuous case. As above, we get an infinite integral (or sum) if $f_x$ fails to be absolutely continuous w.r.t.\ $g_x$ for $x$ on a set of nonzero $p$-measure.
	This quantity is sometimes called the \emph{conditional KL divergence}. 
	\item For Rényi's $\alpha$-divergence, looking directly at the continuous case, we get from \eqref{chainalphaeq} that 
	\begin{align*}
	D_\alpha(f\et g\,|\,p) = D_\alpha(fp\et gp) &= \dfrac{1}{\alpha-1} \, \ln \left( \int_X e^{(\alpha-1) D_\alpha(f_x\et f'_x)} \, p(dx) \right) \\
	&= \dfrac{1}{\alpha-1} \, \ln \left( \int_X \int_Y \left(\dfrac{df_x}{df'_x} \right)^{\alpha-1} df_x \, p(dx) \right) .
	\end{align*}
	Once again, the quantity is infinite if $f_x$ is not absolutely continuous w.r.t.\ $g_x$ for $x$ on a set of positive measure.
	While this quantity does not seem to appear in the literature, we could call it the \emph{conditional Rényi divergence}.\footnote{Other definitions are possible, see \cite{berens} for an analogous discussion on possible definitions of the conditional Rényi entropy.}
	
	\item For the total variation distance, we get once again directly in the continuous case, that 
	\begin{align*}
	d_T(f, g\,|\,p) = d_T(fp, gp) &= \sup_{S\in\Sigma_{X\times Y}}  \big| pf\,(S) - pg\,(S) \big| \\
	&= \sup_{S\in\Sigma_{X\times Y}} \left| \int_X \big( f(S_x|x) - g(S_x|x) \big) \, p(dx) \right| ,
	\end{align*}
	and by equivalently choosing an $S$ such that $f(S_x|x) \ge g(S_x|x)$ for all $x\in X$, we get
	\begin{align*}
		d_T(f, g\,|\,p) &= \sup_{S\in\Sigma_{X\times Y}} \int_X \big| f(S_x|x) - g(S_x|x) \big| \, p(dx) .
	\end{align*}
	Again motivated by the ideas above, we can call this quantity the \emph{conditional total variation distance}. Again, other definitions may be possible. 
\end{itemize}

\section{Measures of stochastic interaction}\label{mi}

In a Markov category, one says \cite[Lemma~12.11 and Definition~12.12]{fritz2019synthetic} that a joint source $h$ on $X$ and $Y$ displays \emph{independence} between $X$ and $Y$ if and only if
\begin{equation}
\begin{tikzpicture}
	\begin{pgfonlayer}{nodelayer}
		\node [style=none] (4) at (2.5, 1.5) {};
		\node [style=none] (5) at (2.5, -1.5) {};
		\node [style=none] (6) at (0, 0) {$=$};
		\node [style=none] (7) at (-4, 0.5) {};
		\node [style=none] (8) at (-4, -0.5) {};
		\node [style=none] (9) at (-2.5, -0.5) {};
		\node [style=none] (10) at (-2.5, 0.5) {};
		\node [style=none] (11) at (-2, 0.5) {$X$};
		\node [style=none] (12) at (-2, -0.5) {$Y$};
		\node [style=none] (13) at (4.5, 1.5) {$X$};
		\node [style=none] (14) at (4.5, -1.5) {$Y$};
		\node [style=none] (16) at (4, 1.5) {};
		\node [style=none] (17) at (2.5, 0.5) {};
		\node [style=none] (18) at (4, -1.5) {};
		\node [style=none] (19) at (2.5, -0.5) {};
		\node [style=bn] (20) at (3.75, 0.5) {};
		\node [style=bn] (21) at (3.75, -0.5) {};
		\node [style=horiz state] (25) at (-4, 0) {$h$};
		\node [style=horiz state] (26) at (2.5, 1) {$h$};
		\node [style=horiz state] (27) at (2.5, -1) {$h$};
	\end{pgfonlayer}
	\begin{pgfonlayer}{edgelayer}
		\draw (18.center) to (5.center);
		\draw (4.center) to (16.center);
		\draw (7.center) to (10.center);
		\draw (8.center) to (9.center);
		\draw (17.center) to (20);
		\draw (19.center) to (21);
	\end{pgfonlayer}
\end{tikzpicture} 
\end{equation}
i.e. if $h$ is the products of its marginals. 
For discrete probability measures, this is exactly the condition 
\[
 p(x,y) = p(x)\,p(y) . 
\]

More generally, a morphism $h:A\to X\otimes Y$ displays \emph{conditional independence} between $X$ and $Y$ if and only if
\begin{equation}\label{indep}
\begin{tikzpicture}
	\begin{pgfonlayer}{nodelayer}
		\node [style=bn] (0) at (4, 0) {};
		\node [style=none] (1) at (5, -1) {};
		\node [style=none] (2) at (5, 1) {};
		\node [style=none] (5) at (7, 1.5) {};
		\node [style=none] (6) at (7, -1.5) {};
		\node [style=none] (7) at (2.75, 0) {};
		\node [style=none] (8) at (0, 0) {$=$};
		\node [style=none] (9) at (-5.5, 0) {};
		\node [style=none] (10) at (-4, 0.5) {};
		\node [style=none] (11) at (-4, -0.5) {};
		\node [style=none] (12) at (-2.5, -0.5) {};
		\node [style=none] (13) at (-2.5, 0.5) {};
		\node [style=none] (14) at (-2, 0.5) {$X$};
		\node [style=none] (15) at (-2, -0.5) {$Y$};
		\node [style=none] (16) at (7.5, 1.5) {$X$};
		\node [style=none] (17) at (7.5, -1.5) {$Y$};
		\node [style=none] (18) at (2, 0) {$A$};
		\node [style=none] (19) at (-6, 0) {$A$};
		\node [style=none] (21) at (5.5, 1.5) {};
		\node [style=none] (22) at (5.5, 0.5) {};
		\node [style=none] (23) at (5.5, -1.5) {};
		\node [style=none] (24) at (5.5, -0.5) {};
		\node [style=bn] (25) at (6.75, 0.5) {};
		\node [style=bn] (26) at (6.75, -0.5) {};
		\node [style=medium box] (27) at (-4, 0) {$h$};
		\node [style=medium box] (28) at (5.5, 1) {$h$};
		\node [style=medium box] (29) at (5.5, -1) {$h$};
	\end{pgfonlayer}
	\begin{pgfonlayer}{edgelayer}
		\draw (7.center) to (0);
		\draw [in=180, out=75] (0) to (2.center);
		\draw (23.center) to (6.center);
		\draw (5.center) to (21.center);
		\draw [in=-75, out=-180] (1.center) to (0);
		\draw (10.center) to (13.center);
		\draw (11.center) to (12.center);
		\draw (22.center) to (25);
		\draw (24.center) to (26);
		\draw (27) to (9.center);
	\end{pgfonlayer}
\end{tikzpicture}
\end{equation}
For discrete probability measures, this is exactly the condition 
\[ 
 p(x,y|a) = p(x|a)\,p(y|a) ,
\]
where copying $A$ corresponds to $a$ appearing twice on the right-hand side of the equation above.  

It is then natural to quantify the stochastic dependence of the variables $X$ and $Y$ by taking the divergence between both sides of the equation. 

\begin{definition}\label{defmi}
 Let $\cat{C}$ be a Markov category with a divergence $D$. The \emph{mutual information} displayed by a morphism $h:A\to X\otimes Y$ is the divergence between the two sides of equation \eqref{indep},
 $$
 \mathrm{I}_D(h) \coloneqq  D\big( h \et (h_X\otimes h_Y)\circ\cop_A \big) .
 $$
 Note that the order of the arguments of $D$ matters.
\end{definition}
 
By construction, every morphism $h$ as above exhibiting (conditional) independence of $X$ and $Y$ (given $A$) has $\mathrm{I}_D(h)=0$.
The converse holds if $D$ is a strict divergence.

\subsection{Data processing inequality}\label{datapmi}

The data processing inequality for a divergence makes these mutual information measures satisfy automatically a data processing inequality.

\begin{proposition}[Data processing inequality for mutual information]\label{dataprocmi}
 Let $f:~A\to X\times Y$, $g:X\to X'$ and $h:Y\to Y'$ be channels in a Markov category with a divergence $D$.
 Then 
 \begin{equation}\label{dataprocpar}
 I_D\big((g\otimes h)\circ f\big) \le I_D(f) .
\end{equation}
\end{proposition}
The situation is as follows,
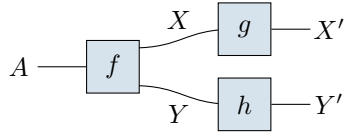
\[
 \begin{tikzpicture}
	\begin{pgfonlayer}{nodelayer}
		\node [style=none] (7) at (-3, 0.5) {};
		\node [style=none] (8) at (-3, -0.5) {};
		\node [style=none] (9) at (-0.25, -1) {};
		\node [style=none] (10) at (-0.25, 1) {};
		\node [style=none] (11) at (-1.75, 1.25) {$X$};
		\node [style=none] (12) at (-1.75, -1.25) {$Y$};
		\node [style=medium box] (26) at (0, 1) {$g$};
		\node [style=medium box] (27) at (0, -1) {$h$};
		\node [style=none] (28) at (2.25, 1) {$X'$};
		\node [style=none] (29) at (2.25, -1) {$Y'$};
		\node [style=none] (30) at (1.75, 1) {};
		\node [style=none] (31) at (1.75, -1) {};
		\node [style=none] (33) at (-5.5, 0) {};
		\node [style=none] (34) at (-6, 0) {$A$};
		\node [style=medium box] (35) at (-3.5, 0) {$f$};
	\end{pgfonlayer}
	\begin{pgfonlayer}{edgelayer}
		\draw [in=-180, out=0, looseness=1.25] (7.center) to (10.center);
		\draw [in=180, out=0, looseness=1.25] (8.center) to (9.center);
		\draw (30.center) to (26);
		\draw (31.center) to (27);
		\draw (35) to (33.center);
	\end{pgfonlayer}
\end{tikzpicture} 
\]
where $X$ and $Y$, coming from a joint source or channel, are processed independently using $g$ and $h$ respectively.

\begin{proof}
 First notice that the first marginal of $(g\otimes h)\circ f$ is $g\circ f_X$, 
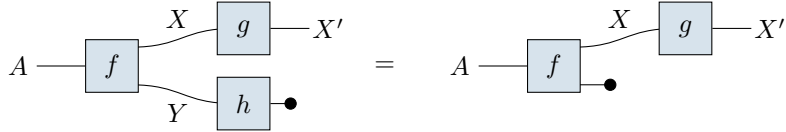
 \[
  \begin{tikzpicture}
	\begin{pgfonlayer}{nodelayer}
		\node [style=none] (7) at (-6.75, 0.5) {};
		\node [style=none] (8) at (-6.75, -0.5) {};
		\node [style=none] (9) at (-4, -1) {};
		\node [style=none] (10) at (-4, 1) {};
		\node [style=none] (11) at (-5.5, 1.25) {$X$};
		\node [style=none] (12) at (-5.5, -1.25) {$Y$};
		\node [style=medium box] (26) at (-3.75, 1) {$g$};
		\node [style=medium box] (27) at (-3.75, -1) {$h$};
		\node [style=none] (28) at (-1.5, 1) {$X'$};
		\node [style=none] (30) at (-2, 1) {};
		\node [style=none] (31) at (-2.5, -1) {};
		\node [style=none] (33) at (-9.25, 0) {};
		\node [style=none] (34) at (-9.75, 0) {$A$};
		\node [style=medium box] (35) at (-7.25, 0) {$f$};
		\node [style=bn] (36) at (-2.5, -1) {};
		\node [style=none] (37) at (5, 0.5) {};
		\node [style=none] (38) at (5, -0.5) {};
		\node [style=none] (39) at (6, -0.5) {};
		\node [style=none] (40) at (7.75, 1) {};
		\node [style=none] (41) at (6.25, 1.25) {$X$};
		\node [style=medium box] (43) at (8, 1) {$g$};
		\node [style=none] (45) at (10.25, 1) {$X'$};
		\node [style=none] (46) at (9.75, 1) {};
		\node [style=none] (48) at (2.5, 0) {};
		\node [style=none] (49) at (2, 0) {$A$};
		\node [style=medium box] (50) at (4.5, 0) {$f$};
		\node [style=bn] (51) at (6, -0.5) {};
		\node [style=none] (52) at (0, 0) {$=$};
	\end{pgfonlayer}
	\begin{pgfonlayer}{edgelayer}
		\draw [in=180, out=0, looseness=1.25] (7.center) to (10.center);
		\draw [in=180, out=0, looseness=1.25] (8.center) to (9.center);
		\draw (30.center) to (26);
		\draw (31.center) to (27);
		\draw (35) to (33.center);
		\draw [in=180, out=0, looseness=1.25] (37.center) to (40.center);
		\draw [in=180, out=0] (38.center) to (39.center);
		\draw (46.center) to (43);
		\draw (50) to (48.center);
	\end{pgfonlayer}
\end{tikzpicture}
 \]
 and similarly the second one is $h\circ f_Y$.
 Now, using \Cref{Ddataprocgen},
 \begin{align*}
 I_D\big((g\otimes h)\circ f\big)  &=  D\big((g\otimes h)\circ f \et (((g\otimes h)\circ f)_X\otimes ((g\otimes h)\circ f)_Y) \circ\cop\big) \\
 &= D\big((g\otimes h)\circ f \et ((g\circ f_X)\otimes (h\circ f_Y))\circ\cop\big) \\
 &= D\big((g\otimes h)\circ f \et (g\otimes h)\circ(f_X\otimes f_Y) \circ\cop\big) \\
 &\le D\big(f \et (f_X\otimes f_Y) \circ\cop\big) \\
 &=  I_D(f) . \qedhere
 \end{align*}
\end{proof}

For sources, in particular, one gets
\[
 I_D((g\otimes h)\circ p) \le I(p) .
\]
In terms of random variables,
\[
 I(g(X):h(Y)) \le I(X:Y) .
\]
This is a quantitative version of the fact that if $X$ and $Y$ are independent, then $g(X)$ and $h(Y)$ are independent too. 
Applying this to marginalizations we get that 
\[
 I(X:Y) \le I(X,A : Y,B) .
\]
i.e.~a monotonicity property in the number of random variables that we are observing.
All the examples that we have satisfy automatically these inequalities.

\subsection{Particular cases}\label{micases}

\subsubsection{Shannon mutual information}

 As it is well known, from the KL divergence on $\cat{FinStoch}$ we get exactly the Shannon mutual information. If $p$ is a discrete probability measure on $X\times Y$,
 \begin{align*}
  \mathrm{I}_{KL}(p)  &=  D_{KL}( p \et p_X\otimes p_Y ) \\
   &= \sum_{x,y} p(x,y) \ln \dfrac{p(x,y)}{p(x)\,p(y)} \\
   &= I(X:Y) .
 \end{align*}
 
 Similarly, for a channel between finite sets $h:A\to X\times Y$, we maximize over the inputs, obtaining the maximal conditional mutual information:
\begin{align*}
 I_{KL}(h) &= \max_{a\in A} \left( \sum_{x,y} h(x,y|a) \ln \dfrac{h(x,y|a)}{h(x|a)\,h(y|a)}  \right) \\
 &= \max_{a\in A} I(X:Y| A=a) .
\end{align*}
 
\begin{remark}
 Since the maximum of an affine function on a convex set is attained on an extremum, we also have 
 \begin{align*}
 I_{KL}(h) &= \max_{p\in PA} \sum_{a,x,y} p(a) \, h(x,y|a) \ln \dfrac{h(x,y|a)}{h(x|a)\,h(y|a)} \\
   &= \max_{p\in PA} I(X:Y | A) ,
 \end{align*}
 the maximum possible conditional mutual information for the joint source $hp$ on $A\times X\times Y$ (maximized over the input source $p$).
\end{remark}

In the nondiscrete case the situation is similar. In particular, when $X$ and $Y$ are standard Borel, 
we can apply the disintegration theorem and express $p$ as the joint $fp_X$ for a regular conditional (kernel) $f:X\to Y$. This way, using \Cref{jointdens},
we can look at whether $f_x\ll p_Y$ for $p_X$-almost all $x$. In this case we can apply the chain rule of entropy, \Cref{KLcontcr}, and we get that
\begin{align*}
 I_{KL}(p) &= D_{KL}( p \et p_X\otimes p_Y ) \\
  &= D_{KL}( fp_X \et p_X\otimes p_Y ) \\
  &= \int_{X} D_{KL}\left( f_x \et p_Y \right) \,p_X(dx) .
\end{align*}
Instead, if $f_x$ fails to absolutely continuous with respect to $p_Y$ on a set of positive $p_X$-measure, we have 
that $I_{KL}(p)=\infty$. 
The same can be done interchanging $X$ and $Y$.

\subsubsection{Rényi or alpha-mutual information}

The Rényi divergence of order $\alpha$, for $\alpha\in (0,\infty), \alpha\ne 1$, gives the following measure of interaction for a discrete probability measure $p$ on $X\times Y$,
\begin{align*}
 I_\alpha(p) &= D_\alpha(p \et p_X\otimes p_Y) \\
  &= \dfrac{1}{\alpha-1} \, \ln \left( \sum_{x,y} \dfrac{p(x,y)^\alpha}{p(x)^{\alpha-1}\,p(y)^{\alpha-1}} \right) ;
\end{align*}

This quantity is sometimes known as \emph{Rényi mutual information} or \emph{$\alpha$-mutual information}~\cite{alphami}.

Let's now look at the nondiscrete case, just as we did for Shannon. When $X$ and $Y$ are standard Borel, 
again by the disintegration theorem, we can write $p$ as the joint $fp_X$ for some kernel $f:X\to Y$. Again, using \Cref{jointdens},
we can look at whether $f_x\ll p_Y$ for $p_X$-almost all $x$. In this case, we can apply the logarithmic chain rule of \Cref{chainalpha}, which gives us
\begin{align*}
 I_{\alpha}(p) &= D_{\alpha}( p \et p_X\otimes p_Y ) \\
  &= D_{\alpha}( fp_X \et p_X\otimes p_Y ) \\
  &= \dfrac{1}{\alpha-1} \, \ln \left( \int_{X} e^{(\alpha-1)\,D_\alpha(f_x \et p_Y)}\, p_X(dx) \right) .
\end{align*}
Instead, if $f_x$ fails to absolutely continuous with respect to $p_Y$ on a set of positive $p_X$-measure, we have 
that $I_{\alpha}(p)=\infty$. 
Once again, the same holds interchanging $X$ and $Y$.

\subsubsection{Total variation mutual information}

The total variation distance gives the following measure of interaction for a discrete probability measure $p$ on $X\times Y$,
\begin{align*}
 I_{T}(p) &= d_t(p, p_X\otimes p_Y) \\
  &= \dfrac{1}{2} \sum_{x,y} \big| p(x,y) - p(x)\,p(y) \big| \\
  &=  \dfrac{1}{2} \sum_{x\in X} p(x) \left( \sum_{y\in Y} \big| p(y|x) - p(y) \big| \right) .
\end{align*}
If we set 
\[
 A_x = \{y\in Y : p(x,y)>p(x)\,p(y)\},
\]
we get
\begin{align*}
 I_T(p) &= \dfrac{1}{2} \sum_{x\in X} p(x) \left( p(A_x|x) - p(A_x) - p(Y\setminus A_x|x) + p(Y\setminus A_x) \right) \\
 &= \sum_{x\in X} p(x) \left( p(A_x|x) - p(A_x) \right) .
\end{align*}
If we moreover set
\[
 B= \{(x,y)\in X\times Y : p(x,y)>p(x)\,p(y)\} = \coprod_{x\in X} \{x\}\times A_x ,
\]
we get exactly
\[
 I_T(p) = p(B) - p_X\otimes p_Y(B) .
\]

In the nondiscrete case, the situation is similar. As before, if $X$ and $Y$ are standard Borel, by the disintegration theorem we can express $p$ as a joint $fp_X$ for some regular conditional $f:X\to Y$, and we get the following.
\begin{align*}
 I_{T}(p) &= d_t(p, p_X\otimes p_Y) \\
 &= \sup_{S\in\Sigma_{X\times Y}} \big| p(S) - p_X\otimes p_X\,(S) \big| \\
 &= \sup_{S\in\Sigma_{X\times Y}} \int_X \big| f(S_x|x) - p_Y(S_x) \big| \, p_X(dx) ,
\end{align*}
where again we see a comparison of the conditional with the marginal.

\subsection{Conditional mutual information}\label{condmi}

Following the guidelines of \Cref{conddiv}, we can also measure the departure from \emph{almost sure} independence. 
Given a channel $h:A\to X\otimes Y$ and a source $p$ on $A$, we say that $h$ exibits conditional independence of $X$ and $Y$ $p$-almost surely if and only if equation \eqref{indep} holds $p$-almost surely, that is, the following equation holds.
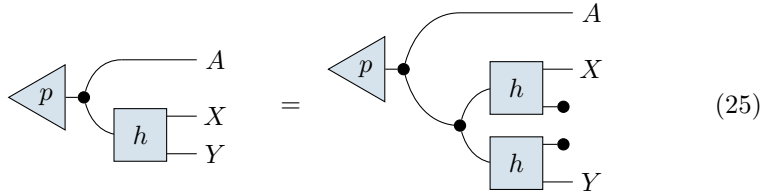
\begin{equation}\label{asind}
	\begin{tikzpicture}
	\begin{pgfonlayer}{nodelayer}
		\node [style=bn] (0) at (4.5, -0.5) {};
		\node [style=none] (1) at (5.5, -1.5) {};
		\node [style=none] (2) at (5.5, 0.5) {};
		\node [style=none] (5) at (7.5, 1) {};
		\node [style=none] (6) at (7.5, -2) {};
		\node [style=none] (8) at (0, 0) {$=$};
		\node [style=none] (10) at (-4, -0.25) {};
		\node [style=none] (11) at (-4, -1.25) {};
		\node [style=none] (12) at (-2.5, -1.25) {};
		\node [style=none] (13) at (-2.5, -0.25) {};
		\node [style=none] (14) at (-2, -0.25) {$X$};
		\node [style=none] (15) at (-2, -1.25) {$Y$};
		\node [style=none] (16) at (8, 1) {$X$};
		\node [style=none] (17) at (8, -2) {$Y$};
		\node [style=none] (19) at (-2, 1.25) {$A$};
		\node [style=none] (21) at (6, 1) {};
		\node [style=none] (22) at (6, 0) {};
		\node [style=none] (23) at (6, -2) {};
		\node [style=none] (24) at (6, -1) {};
		\node [style=bn] (25) at (7.25, 0) {};
		\node [style=bn] (26) at (7.25, -1) {};
		\node [style=medium box] (27) at (-4, -0.75) {$h$};
		\node [style=medium box] (28) at (6, 0.5) {$h$};
		\node [style=medium box] (29) at (6, -1.5) {$h$};
		\node [style=horiz state] (30) at (-6.5, 0.25) {$p$};
		\node [style=none] (31) at (-4.5, -0.75) {};
		\node [style=none] (33) at (-4.5, 1.25) {};
		\node [style=none] (34) at (-2.5, 1.25) {};
		\node [style=bn] (35) at (-5.5, 0.25) {};
		\node [style=none] (36) at (8, 2.5) {$A$};
		\node [style=horiz state] (37) at (2, 1) {$p$};
		\node [style=none] (38) at (4.5, -0.5) {};
		\node [style=none] (39) at (4.5, 2.5) {};
		\node [style=none] (40) at (7.5, 2.5) {};
		\node [style=bn] (41) at (3, 1) {};
	\end{pgfonlayer}
	\begin{pgfonlayer}{edgelayer}
		\draw [in=180, out=75] (0) to (2.center);
		\draw (23.center) to (6.center);
		\draw (5.center) to (21.center);
		\draw [in=-75, out=-180] (1.center) to (0);
		\draw (10.center) to (13.center);
		\draw (11.center) to (12.center);
		\draw (22.center) to (25);
		\draw (24.center) to (26);
		\draw (34.center) to (33.center);
		\draw (30) to (35);
		\draw [in=180, out=75] (35) to (33.center);
		\draw [in=180, out=-75] (35) to (31.center);
		\draw (40.center) to (39.center);
		\draw (37) to (41);
		\draw [in=180, out=75] (41) to (39.center);
		\draw [in=180, out=-75] (41) to (38.center);
		\draw (38.center) to (0);
	\end{pgfonlayer}
\end{tikzpicture}
\end{equation}
For discrete probability measures, this means that 
$$
p(a)\,p(x,y|a) = p(a)\,p(x|a)\,p(y|a) ,
$$
which means that $p(x,y|a) = p(x|a)\,p(y|a)$ for all the $a$ of nonzero probability. The continuous case is analogous. 

We can take the divergence between both sides of \eqref{asind}, and call the resulting quantity the \emph{conditional mutual information}, denoted by $I_D(h|p)$. Let's see this in our examples for the discrete case.
\begin{itemize}
	\item For the KL divergence, we get exactly Shannon's conditional mutual information:
	$$
	I_{KL}(h|p) = \sum_{a,x,y} p(a) \, h(x,y|a) \ln \dfrac{h(x,y|a)}{h(x|a)\,h(y|a)} .
	$$
	(In the continuous case the situation is analogous, with the usual infinities if absolute continuity fails $p$-almost surely.)
	
	\item For the Rényi divergences, we get the following quantity:
	$$
	I_\alpha(h|p) = \dfrac{1}{\alpha-1} \ln \left( \sum_{a,x,y} p(a) \,\dfrac{h(x,y|a)^\alpha}{h(x|a)^{\alpha-1}\,h(y|a)^{\alpha-1}} \right) ,
	$$ 
	which analogously we can call the \emph{Rényi conditional mutual information}. 
	
	\item For the total variation distance, 
	$$
	I_T(h|p) = \dfrac{1}{2} \sum_{a,x,y} p(a)\,\big| h(x,y|a) - h(x|a)\,h(y|a) \big| .
	$$
\end{itemize}

\section{Measures of randomness}\label{H}

In a Markov category, one says \cite[Definition~10.1]{fritz2019synthetic} that a morphism $f:~X\to Y$ in a Markov category is \emph{deterministic} if and only if applying $f$ and copying its output is the same as copying its input and applying it to both copies:
\begin{equation}\label{det} 
\begin{tikzpicture}
	\begin{pgfonlayer}{nodelayer}
		\node [style=none] (7) at (-6.5, 0) {};
		\node [style=bn] (8) at (-3.5, 0) {};
		\node [style=none] (9) at (-2.5, 1) {};
		\node [style=none] (11) at (-2.5, -1) {};
		\node [style=none] (16) at (0, 0) {$=$};
		\node [style=none] (20) at (-2, -1) {$Y$};
		\node [style=none] (21) at (-2, 1) {$Y$};
		\node [style=none] (23) at (-7, 0) {$X$};
		\node [style=medium box] (24) at (-5, 0) {$f$};
		\node [style=none] (25) at (2.5, 0) {};
		\node [style=bn] (26) at (3.5, 0) {};
		\node [style=none] (27) at (4.5, 1) {};
		\node [style=none] (28) at (4.5, -1) {};
		\node [style=none] (29) at (7, -1) {$Y$};
		\node [style=none] (30) at (7, 1) {$Y$};
		\node [style=none] (31) at (2, 0) {$X$};
		\node [style=medium box] (32) at (5, 1) {$f$};
		\node [style=none] (33) at (6.5, 1) {};
		\node [style=none] (34) at (6.5, -1) {};
		\node [style=medium box] (35) at (5, -1) {$f$};
	\end{pgfonlayer}
	\begin{pgfonlayer}{edgelayer}
		\draw [style=none] (7.center) to (8);
		\draw [style=none, in=180, out=75] (8) to (9.center);
		\draw [in=180, out=-75] (8) to (11.center);
		\draw [style=none] (25.center) to (26);
		\draw [style=none, in=180, out=75] (26) to (27.center);
		\draw [in=180, out=-75] (26) to (28.center);
		\draw (27.center) to (33.center);
		\draw (34.center) to (28.center);
	\end{pgfonlayer}
\end{tikzpicture}
\end{equation} 
The intuition is that, if $f$ involves some randomness, the left-hand side of the equation will display perfectly correlated results, possibly noisy, but guaranteed to be equal. 
On the right-hand side, instead, the two functions are executed independently, and so the outputs will be completely independent (technically speaking, conditionally independent given the input).
If $f$ is just a function, both sides of the equation map an input $x$ simply to $(f(x), f(x))$. 

For sources, in particular, the equation reduces to the following.
\begin{equation}
\begin{tikzpicture}
	\begin{pgfonlayer}{nodelayer}
		\node [style=none] (7) at (-5, 0) {};
		\node [style=bn] (8) at (-3.5, 0) {};
		\node [style=none] (9) at (-2.5, 1) {};
		\node [style=none] (11) at (-2.5, -1) {};
		\node [style=none] (16) at (0, 0) {$=$};
		\node [style=none] (20) at (-2, -1) {$X$};
		\node [style=none] (21) at (-2, 1) {$X$};
		\node [style=none] (27) at (2.5, 1) {};
		\node [style=none] (28) at (2.5, -1) {};
		\node [style=none] (29) at (4.5, -1) {$X$};
		\node [style=none] (30) at (4.5, 1) {$X$};
		\node [style=none] (33) at (4, 1) {};
		\node [style=none] (34) at (4, -1) {};
		\node [style=horiz state] (36) at (-5, 0) {$p$};
		\node [style=horiz state] (37) at (2.5, 1) {$p$};
		\node [style=horiz state] (38) at (2.5, -1) {$p$};
	\end{pgfonlayer}
	\begin{pgfonlayer}{edgelayer}
		\draw [style=none] (7.center) to (8);
		\draw [style=none, in=180, out=75] (8) to (9.center);
		\draw [in=180, out=-75] (8) to (11.center);
		\draw (27.center) to (33.center);
		\draw (34.center) to (28.center);
	\end{pgfonlayer}
\end{tikzpicture}
\end{equation} 
In $\cat{Stoch}$ and $\cat{FinStoch}$, the left-hand side is a measure supported on the diagonal of $X\times X$, and the right-hand side is a product measure.
Recalling the idea of independence and conditional independence (equation \eqref{indep}), a deterministic source can be seen as a source that is independent of itself, and a deterministic morphism is a morphism which is conditionally independent of itself given its input. This is compatible with the usual interpretation in probability theory.
In $\cat{FinStoch}$, deterministic morphisms are precisely the stochastic matrices whose entries are zero and one, i.e.~those defined by (deterministic) functions. 
In $\cat{Stoch}$, deterministic measures are precisely those that give probability either zero or one to each event.\footnote{Such measures are sometimes said to satisfy a \emph{zero-one law}, more on that in \cite{fritzrischel2019zeroone}.} 
For standard Borel spaces, those are exactly the Dirac delta measures. For more general measurable spaces, there are measures which are zero-one, but which are not Dirac deltas (for example ergodic, measures on invariant $\sigma$-algebras, see \cite{ergodic}). 
The situation is exactly the same for kernels. 

We define as our measure of randomness the discrepancy between the two sides of equation \eqref{det}.

\begin{definition}\label{defentropy} 
 Let $\cat{C}$ be a Markov category with divergence $D$. The \emph{entropy} of a morphism $f:X\to Y$ is the quantity
 \begin{equation}
  H_D(f) \coloneqq D \big( \cop \circ f \et (f\otimes f)\circ \cop \big) ,
 \end{equation}
 i.e.~the divergence between the two sides of \eqref{det}. (Note that the order matters.)
\end{definition}

In other words, entropy is the mutual information of the left-hand side of \eqref{det}. 
This corresponds to the usual identity for Shannon entropy and mutual information,
\[
 H(X) = I(X:X) .
\]

By construction, every deterministic morphism has zero entropy. 
If the divergence $D$ is strict, we have that conversely a morphism of zero entropy is deterministic.

\subsection{Data processing inequality}

Defined in this way, entropy automatically satisfies a data processing inequality.

\begin{proposition}[Data processing inequality for entropy]\label{dataprocH}
 Let $\cat{C}$ be a Markov category with a divergence $D$, and consider morphisms $f:X\to Y$ and $g:~Y\to Z$, with $g$ deterministic.
 Then
 \[
   H_D(g\circ f ) \le H_D(f) .
 \]
\end{proposition}

In particular, for sources,
\[
 H_D(g\circ p ) \le H_D(p) .
\]
Note that this requires the morphism $g$ to be deterministic, otherwise additional randomness might be added.

\begin{proof}[Proof of \Cref{dataprocH}.]
 Note that since $g$ is deterministic, by definition it satisfies \eqref{det}. 
 Now by \Cref{dataprocpar},
 \begin{align*}
  H_D(g\circ f) &= I_D(\cop \circ g\circ f) \\
  &=  I_D\big((g\otimes g)\circ\cop\circ f\big) \\
  &\le I_D\big(\cop\circ f\big) \\
  &= H_D(f) . \qedhere
 \end{align*}
\end{proof}

For sources, and in terms of random variables, this inequality reads as 
\[
 H_D(g(X)) \le H_D(X) .
\]
In particular, since marginalizations are deterministic, this implies the monotonicity property,
\[
 H_D(X) \le H_D(X,Y) ,
\]
i.e.~that entropy increases with the number of variables observed, as it happens with Shannon entropy.

\subsection{Particular entropies in the finite case}\label{Hcasefin}

\subsubsection{Shannon entropy}

The relative entropy divergence on $\cat{FinStoch}$ gives exactly the Shannon entropy:
 \begin{align*}
  H (p) &= D_{KL} \big( \cop\circ\,p \et p\otimes p \big) \\ 
   &= \sum_{x,x'\in X} p(x)\,\delta_{x,x'} \ln \dfrac{ p(x)\,\delta_{x,x'}}{p(x)\,p(x')} \\
   &= \sum_{x\in X} p(x) \ln \dfrac{ p(x)}{p(x) \,p(x)} \\
   &= \sum_{x\in X} p(x) \ln \dfrac{1}{p(x)} \\
   &= - \sum_{x\in X} p(x) \ln p(x) ,
 \end{align*}
 recalling \Cref{logzero}. 
 
Similarly, for a channel between finite sets $f:X\to Y$, we maximize over the inputs,
\[
 H (f) = \max_{x\in X} \left( - \sum_{y} f(y|x) \ln f(y|x) \right) .
\]

\begin{remark}
 Since the maximum of an affine function on a convex set is attained on an extremum, we also have that
 \begin{align*}
  H (f) 
   &= \max_{x\in X} \left( - \sum_{y} f(y|x) \ln f(y|x) \right) \\
   &= \max_{p\in PX} \left( - \sum_{y,x} p(x) f(y|x) \ln f(y|x) \right) \\
   &= \max_{p\in PX} H_{pf}(Y|X) ,
 \end{align*}
 i.e.~the maximum possible entropy of the channel, maximized over the inputs. (Conditional entropy is not to be confused with relative entropy.)
\end{remark}

\subsubsection{Rényi entropy}

For $\alpha\ne 1$,
\begin{align*}
 D_\alpha(\cop\circ\,p \et p\otimes p) &= \dfrac{1}{\alpha-1} \, \ln \left( \sum_{x,x'\in X} \dfrac{\big(p(x)\,\delta_{x,x'}\big)^\alpha}{\big(p(x)\,p(x')\big)^{\alpha-1}} \right) \\
  &= \dfrac{1}{\alpha-1} \, \ln \left( \sum_{x\in X} \dfrac{p(x)^\alpha}{\big(p(x)\,p(x)\big)^{\alpha-1}} \right) \\
  &= \dfrac{1}{\alpha-1} \, \ln \left( \sum_{x\in X} p(x)^{2-\alpha} \right) .
\end{align*}

This quantity is known as the \emph{Rényi entropy}~\cite{renyientropy}, but \emph{of a different order}: namely, usually the Rényi entropy of order $\alpha$ is usually defined as 
\[
 H_\alpha(p) = \dfrac{1}{1-\alpha} \, \ln \left( \sum_{x\in X} p(x)^\alpha \right) .
\]
Therefore, 
\[
 H_{D_\alpha}(p) = H_{2-\alpha}(p) .
\]
Notice that the factor in front of the logarithm does not need a correction, and that the case $\alpha=1$ is still saying that the KL divergence gives the Shannon entropy.
This fact is known at least since \cite{renyi-div}, see the end of section I-A therein.

\subsubsection{Total variation and the Gini-Simpson index (linear entropy)}\label{gini}

The entropy measure defined by the total variation distance for a probability measure $p$ on $X$ reads as follows. 
\[
 H_T(p) = d_T(\cop\circ\,p, p\otimes p) = \sup_{S\in\Sigma_{X\times X}} \big| \copy\circ p\,(S) - p\otimes p\,(S) \big| .
\]

If $p$ is a discrete probability measure, the entropy reduces to the following quantity:
\begin{align*}
 H_T(p) &= \dfrac{1}{2} \sum_{x,y\in X} \big| p(x)\,p(y) - \delta_{xy}\,p(x) \big| \\
  &= \dfrac{1}{2} \sum_{x\in X} p(x) \sum_{y\in Y} \big| p(y) - \delta_{xy} \big| \\
  &= \dfrac{1}{2} \sum_{x\in X} p(x) \left( 1-p(x) + \sum_{y\ne x} p(y) \right) \\
  &= \dfrac{1}{2} \sum_{x\in X} p(x) \big( 1-p(x)+ 1-p(x) \big) \\
  &= \sum_{x\in X} p(x) \big( 1-p(x) \big) \\
  &= 1 - \sum_{x\in X} p(x)^2 .
\end{align*}

This quantity is called \emph{Gini-Simpson index} \cite[Example~4.1.3.iii]{leinster-entropy}, and is used in ecology to quantify the diversity in an ecosystem (see for example~\cite{loujost}). (This is related to, but different from, the \emph{Gini coefficient} used to measure income inequality in economics.)
This can also be seen as a classical counterpart of \emph{linear entropy} in quantum information theory (see \cite{linearentropy} and references therein). 

The next-to-last line of the calculations above can be interpreted as the probability that any two points drawn independently from the probability distribution $p$ are different. This will be the case also outside the discrete case, see \Cref{gini-cont}.

\begin{question}
 The Gini-Simpson index happens to be the Tsallis entropy of order 2 (see~\emph{\cite{leinster-entropy}}). 
 Can we deform the total variation distance to a family of divergences on $\cat{FinStoch}$ which give us the Tsallis entropy of different orders?
 (Recall that the $q$-divergences do not in general give divergences on $\cat{FinStoch}$, see \Cref{nonexample}.)
\end{question}

\subsection{Entropy for nondiscrete distributions}\label{Hcaseinf}

Let's now study the infinite, nondiscrete case. For general measures, entropy tends to be maximal whenever a measure has no atoms.
The idea is such measures are ``maximally spread''.

Let's see this in detail for the case of standard Borel spaces. 
Recall that given a finite measure $p$ on $X$, an \emph{atom} of $p$ is a measurable subset $A\subseteq X$ with $p(A)>0$ such that if any subset $A'\subseteq A$ has measure $p(A') < p(A)$ strictly, then $p(A')=0$.
If $X$ is standard Borel, any atom is necessarily a singleton $\{x\}$ for some $x\in X$.
We call $p$ \emph{purely atomic} if every set of positive measure contains at least an atom. 
On standard Borel spaces, this is equivalent to say that $p$ is \emph{discrete}, i.e.~it is a countable sum of Dirac deltas. 
At the other extreme, we call $p$ \emph{atomless} if it has no atoms. The Lebesgue (uniform) measure on the unit interval $[0,1]$ is an example of a nonatomic probability measure, the normal distribution on $\R$ is another one. 
Note that there are measures which are neither purely atomic nor atomless, for example a nontrivial convex combination of an atomless distribution and a Dirac delta. 

Recall that on a standard Borel space $X$, the ``diagonal'' set
\[
 \Delta \coloneqq \{(x,y) : x=y\} \subseteq X\times X
\]
is measurable. (This in turn implies that all singletons are measurable.)

\begin{lemma}\label{atomless}
 Let $p$ be a finite measure on a standard Borel space $X$. 
 Then $p\otimes p\,(\Delta) = 0$ if and only if $p$ is atomless.
\end{lemma}
\begin{proof}[Proof of \Cref{atomless}]
 We have, by definition of product measure, that
 \[
  p\otimes p\,(\Delta) = \int_X p(\Delta_x)\, p(dx) .
 \]
 Notice now that 
 \[
  \Delta_x = \{y\in X : (x,y) \in\Delta\} = \{x\} ,
 \]
 so that we are left with 
 \begin{equation}\label{ppdelta}
  p\otimes p\,(\Delta) = \int_X p(\{x\})\, p(dx) .
 \end{equation}
 Now, $p$ is atomless if and only if the measurable function 
 \[
  \begin{tikzcd}[row sep=0]
   X \ar{r}{\tilde p} & {[0,1]} \\
   x \ar[mapsto]{r} & p(\{x\}) 
  \end{tikzcd}
 \]
 is identically zero. If this is the case, \eqref{ppdelta} is zero. Conversely, if \eqref{ppdelta} is zero, it means that $\tilde p(x) = p(\{x\})$ can only be nonzero on a set of $p$-measure zero, but this implies that $\tilde p$ is zero identically.
\end{proof}

Let's now turn to the examples.

\subsubsection{Shannon and Rényi entropies, standard Borel case}

\begin{theorem}\label{abscatomic}
 Let $p$ be a probability distribution on a standard Borel space $X$. Then 
 \[
  \cop\circ\,p \ll p\otimes p
 \]
 if and only if $p$ is discrete.
\end{theorem}
\begin{proof}[Proof of \Cref{abscatomic}]
 First, suppose that $\cop\circ\,p \ll p\otimes p$. 
 Denote by $A(p)$ the set of atoms of $p$, which is countable (and possibly empty). 
 Notice that 
 \[
  A(p\otimes p) = A(p)\times A(p)\subseteq X\times X ,
 \]
 and that 
 \begin{equation}\label{atomdiag}
  A(\cop\circ\,p)=\big(A(p)\times A(p)\big)\cap \Delta = A(p\otimes p) \cap \Delta.
 \end{equation} 
 We now decompose $p$ into a discrete and an atomless part as follows,
 \[
  p = p_{\mbox{atomless}} + p_{\mbox{discrete}}
 \]
 where for each measurable set $S\subseteq X$, we set 
 \[
  p_{\mbox{discrete}}(S) \coloneqq p\big(S\cap A(p) \big)  \quad\mbox{and}\quad
  p_{\mbox{atomless}}(S) \coloneqq p\big(S\setminus A(p) \big) . 
 \]
 By \Cref{atomless}, we have that 
 \[
  p_{\mbox{atomless}}\otimes p_{\mbox{atomless}}\,(\Delta) = 0 .
 \]
 Rewriting the left-hand side in terms of $p$,
 \[
  p\otimes p\,\big(\Delta \setminus A(p\otimes p) \big) = 0 .
 \]
 Now since $\cop\circ\,p \ll p\otimes p$, and using \eqref{atomdiag},
 \[
  \cop\circ\,p\,\big(\Delta \setminus A(p\otimes p) \big) = \cop\circ\,p\,\big(\Delta \setminus A(\cop\circ\,p) \big) = 0 ,
 \]
 i.e.
 \[
  \cop\circ\,p_{\mbox{atomless}}\,\big(\Delta \big) = 0 .
 \]
 Since $\cop\circ\,p\,(\Delta)=1$, we must have that $p$ is discrete.
 
 Conversely, suppose that $p$ is discrete. Then 
 \[
  p = \sum_i \delta_{x_i}
 \]
 for countably many $x_i$. 
 This way,
 \[
  \cop\circ\,p = \sum_{i} \delta_{(x_i,x_i)} \quad\mbox{and}\quad p\otimes p = \sum_{i,j} \delta_{(x_i,x_j)} .
 \]
 Every measurable set $S\subseteq X\otimes X$ has $p\otimes p$-measure zero if and only if it does not contain any of the ordered pairs $(x_i,x_j)$, and so in particular it does not contain the ``diagonal'' ones in the form $(x_i,x_i)$, so that $S$ must also have $\cop\circ\,p$-measure zero.
 Therefore $\cop\circ\,p \ll p\otimes p$.
\end{proof}

\begin{corollary}
 Let $p$ be a probability measure on a standard Borel space $X$. We have that
 \[
  H_{KL}(p) = \begin{cases} 
               - \sum_{x\in A(p)} p(x) \ln p(x) & \mbox{if } p \mbox{ is discrete;} \\
               \infty & \mbox{otherwise.}
              \end{cases}
 \]
\end{corollary}

In particular, for a generic probability measure on $\R$ we do not get differential entropy (see \cite[Chapter~8]{cover-thomas} for the definition). 
Note that differential entropy cannot be obtained from any divergence: it can be negative, and it is nonzero (but instead, negative infinity) on Dirac deltas. 
In the literature it is well known that differential entropy differs from the limit of discrete entropy by an infinite constant \cite[Section~8.3]{cover-thomas}. The entropy $H_{KL}$ that we find here corresponds to the limiting discrete entropy, rather than to differential entropy. 

\begin{question}
 Can we obtain entropy as a supremum over countable partitions, as we saw for the divergence in \Cref{partitions}?
\end{question}

We get a similar result for the Rényi divergence.
\begin{corollary}
 Let $p$ be a probability measure on a standard Borel space $X$. For $\alpha\in(0,1),\alpha\ne 1$,
 \[
  H_{\alpha}(p) = \begin{cases} 
               \dfrac{1}{1-\alpha} \, \ln \left( \sum_{x\in X} p(x)^{\alpha} \right) & \mbox{if } p \mbox{ is discrete;} \\
               \infty & \mbox{otherwise,}
              \end{cases}
 \]
 recalling that the Rényi entropy of order $\alpha$ is given by the divergence of order $2-\alpha$. 
 For $\alpha=0$ and $\alpha=\infty$ one can again take the limit. 
\end{corollary}

\subsubsection{Gini-Simpson index (linear entropy), standard Borel case}\label{gini-cont}

For the total variation distance, we have a similar situation to the discrete case of \Cref{gini}.

\begin{theorem}\label{ginisimpson}
 Let $X$ be a standard Borel space, and denote by $\Delta\subseteq X\times X$ the diagonal subset.
 For each probability measure $p$ on $X$, we have
 \[
  H_T(p) = 1 - p\otimes p \,(\Delta) ,
 \]
 i.e.~the probability that two points drawn independently from the distribution $p$ are not equal.
\end{theorem}

Let's prove the theorem using the following auxiliary statement.

\begin{lemma}\label{diagmore}
 Let $T$ be a measurable subset of $\Delta$. Then 
 \[
  \cop\circ\,p\,(T) \ge p\otimes p\, (T) .
 \]
\end{lemma}
\begin{proof}[Proof of \Cref{diagmore}.]
 Denote by $T'\subseteq X$ the projection $\pi_1(T)\subseteq X$ of $T$ onto its first coordinate.
 Since both measures have $p$ as first (and second) marginal, we have that 
 \[
  \cop\circ\,p\,(T'\times X) = p(T') = p\otimes p\,(T'\times X) .
 \]
 Now $T$ is a measurable subset of $T'\times X$, and so we have 
 \[
  p\otimes p \,(T) \le p\otimes p\,(T'\times X) = \cop\circ\,p\,(T'\times X) .
 \]
 On the other hand,
 \begin{align*}
  \cop\circ\,p \,(T'\times X) &= p(\{x\in X : (x,x)\in T'\times X\}) \\
   &= p(\{x\in X : (x,x)\in T\} \\
   &= \cop\circ\,p(T) . \qedhere
 \end{align*}
\end{proof}

We are now ready to prove the theorem.

\begin{proof}[Proof of \Cref{ginisimpson}.]
 First of all, 
 \begin{equation}\label{suphere}
  H_T(p) = \sup_{S\in\Sigma_{X\times X}} \big| \cop\circ\,p\,(S) - p\otimes p\,(S) \big| .
 \end{equation}
 Using \Cref{diagmore}, we have that 
 \begin{align*}
  H_T(p) &\ge \big| \cop\circ\,p\,(\Delta) - p\otimes p\,(\Delta) \big| \\
   &= \cop\circ\,p\,(\Delta) - p\otimes p\,(\Delta) \\
   &= 1 - p\otimes p\,(\Delta) .
 \end{align*}
 
 Let's now show that $\Delta$ (or equivalently, its complement) maximizes \eqref{suphere}. Let $S$ be any (other) measurable subset of $X\times X$. Then 
 \begin{align*}
  \cop\circ\,p\,(S) - p\otimes p\,(S) &= \begin{multlined}[t]
                                         \cop\circ\,p\,(\Delta\cap S) - p\otimes p\,(\Delta\cap S) \\
                                           + \cop\circ\,p\,(S\setminus\Delta) - p\otimes p\,(S\setminus\Delta)
                                        \end{multlined} \\
 &\le \cop\circ\,p\,(\Delta\cap S) - p\otimes p\,(\Delta\cap S) ,
 \end{align*}
 since $\cop\circ\,p\,(S\setminus\Delta)=0$ and $p\otimes p\,(S\setminus\Delta)\ge 0$. 
 Therefore, in the supremum \eqref{suphere}, we can equivalently restrict to $\Delta$ and its subsets. 
 
 So let $T$ be a measurable subset of $\Delta$. We have that 
 \[
  \cop\circ\,p\,(T) - p\otimes p\,(T) = \begin{multlined}[t]
                                         \cop\circ\,p\,(\Delta) - p\otimes p\,(\Delta) \\
                                           + \cop\circ\,p\,(\Delta\setminus T) - p\otimes p\,(\Delta\setminus T) ,
                                        \end{multlined}
 \]
 but the last line is nonnegative by \Cref{diagmore}, and so the optimum is attained at $\Delta$. 
\end{proof}

\begin{corollary}
 Let $X$ be a standard Borel space. Then for each probability measure $p$ on $X$,
 \[
  H_T(p) = 1 - \sum_{x\in A(p)} p(x)^2 ,
 \]
 where $A(p)$ is the set of atoms of $p$.
 In particular, $H_T(p)=1$ if and only if $p$ is atomless.
\end{corollary}

This generalizes the formula for the Gini-Simpson index that we found in the discrete case.

\subsection{Conditional entropy}\label{condent}

Given a source $p$ on $X$, we say that a channel $f:X\to Y$ is $p$-almost surely deterministic if equation \eqref{det} holds $p$-almost surely, i.e.\ if the following equation holds.
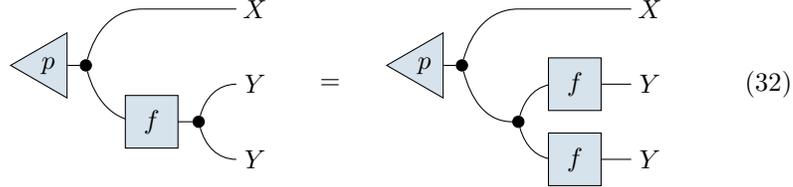
\begin{equation}\label{asdet}
	\begin{tikzpicture}
	\begin{pgfonlayer}{nodelayer}
		\node [style=bn] (8) at (-3.5, -1) {};
		\node [style=none] (9) at (-2.5, 0) {};
		\node [style=none] (11) at (-2.5, -2) {};
		\node [style=none] (16) at (0, 0) {$=$};
		\node [style=none] (20) at (-2, -2) {$Y$};
		\node [style=none] (21) at (-2, 0) {$Y$};
		\node [style=medium box] (24) at (-4.75, -1) {$f$};
		\node [style=bn] (26) at (5, -1) {};
		\node [style=none] (27) at (6, 0) {};
		\node [style=none] (28) at (6, -2) {};
		\node [style=none] (29) at (8.5, -2) {$Y$};
		\node [style=none] (30) at (8.5, 0) {$Y$};
		\node [style=medium box] (32) at (6.5, 0) {$f$};
		\node [style=none] (33) at (8, 0) {};
		\node [style=none] (34) at (8, -2) {};
		\node [style=medium box] (35) at (6.5, -2) {$f$};
		\node [style=none] (37) at (-5, 2) {};
		\node [style=none] (40) at (-2.5, 2) {};
		\node [style=none] (41) at (-2, 2) {$X$};
		\node [style=none] (43) at (5, 2) {};
		\node [style=bn] (45) at (3.5, 0.5) {};
		\node [style=none] (46) at (8, 2) {};
		\node [style=bn] (48) at (-6.5, 0.5) {};
		\node [style=none] (49) at (-5, -1) {};
		\node [style=none] (51) at (8.5, 2) {$X$};
		\node [style=horiz state] (52) at (-7.5, 0.5) {$p$};
		\node [style=horiz state] (53) at (2.5, 0.5) {$p$};
	\end{pgfonlayer}
	\begin{pgfonlayer}{edgelayer}
		\draw [style=none, in=180, out=75] (8) to (9.center);
		\draw [in=180, out=-75] (8) to (11.center);
		\draw [style=none, in=180, out=75] (26) to (27.center);
		\draw [in=180, out=-75] (26) to (28.center);
		\draw (27.center) to (33.center);
		\draw (34.center) to (28.center);
		\draw (40.center) to (37.center);
		\draw [in=75, out=-180] (43.center) to (45);
		\draw (46.center) to (43.center);
		\draw [in=180, out=-75] (45) to (26);
		\draw [in=-180, out=75] (48) to (37.center);
		\draw [in=180, out=-75] (48) to (49.center);
		\draw (49.center) to (8);
		\draw (52) to (48);
		\draw (53) to (45);
	\end{pgfonlayer}
\end{tikzpicture}
\end{equation}

We can take the divergence between both sides of \eqref{asdet}, and call the resulting quantity the \emph{conditional entropy}, denoted by $H_D(f|p)$. Let's see this in our examples, once again for the discrete case (the continuous case is analogous).
\begin{itemize}
	\item For the KL divergence, we get exactly Shannon's conditional entropy:
	$$
	H_{KL}(f|p) = \sum_{x,y} p(x) \, f(y|x) \ln f(y|x) .
	$$
	
	\item For the Rényi divergences, we get the following quantity:
	$$
	H_\alpha(f|p) = \dfrac{1}{\alpha-1} \ln \left( \sum_{x,y} p(x)\,f^{2-\alpha} \right) ,
	$$ 
	which analogously we can call the \emph{Rényi conditional entropy}. This agrees (once again up to $\alpha\mapsto 2-\alpha$) with the definition of conditional Rényi entropy given for example in \cite[Definition~7]{sharplowerbounds}. (As we remarked before, this is not the only possible definition, see \cite{berens} for more.)
	
	\item For the total variation distance, 
	\begin{align*}
	H_T(f|p) &= \sum_x p(x) \left( 1- \sum_yf(y|x)^2 \right) \\
	&= 1 - \sum_{x,y} p(a)\,f(y|x)^2 .
	\end{align*}
	We can call this the \emph{conditional linear entropy} or \emph{conditional Gini-Simpson coefficient}.
\end{itemize}

\subsection{Future work: beyond measurable spaces}\label{futurework}

Entropy, as constructed in this work, is an invariant of measurable spaces, or more generally of the Markov category that one is considering.
The same can be said about divergences. 
However, very often in information theory one uses more structure than just $\sigma$-algebras. In $\R$, say, one also uses the order, the metric, and so on, and divergences and entropies on the Markov category $\cat{Stoch}$ are not detecting them.
For example, consider two points $x,y\in\R$ with distance $|x-y|=\varepsilon>0$. 
No matter how small $\varepsilon$ is, for Dirac delta distributions at $x$ and $y$ we always have 
\[
D_{\alpha}(\delta_x\et\delta_y) = \infty, \qquad d_T(\delta_x,\delta_y) = 1 .
\]
Similarly, given a delta at $0$ and normal distributions centered at zero with variances $1$ and $100$,
we have that 
\[
D_{\alpha}(\delta_0\et N_{0,1}) = \infty, \qquad d_T(\delta_0,N_{0,1}) = 1 ,
\]
and in exactly the same way,
\[
D_{\alpha}(\delta_0\et N_{0,100}) = \infty, \qquad d_T(\delta_0,N_{0,100}) = 1 .
\]
As a consequence, the two normal distributions have the same entropy:
\[
H_\alpha(N_{0,1}) = \infty = H_\alpha(N_{0,100}), \qquad H_T(N_{0,1}) = 1 = H_T(N_{0,100}) .
\]

This is not a feature of the choice of divergence ($D_\alpha,d_T$, and so on), but rather, of the \emph{category}, namely $\cat{Stoch}$. While for discrete probability distributions $\cat{FinStoch}$ and its countable analogue are good models, once we consider infinite alphabets (for example, in $\R$), a measurable structure is not enough. 
To see this, let $p$ and $q$ be atomless probability measures on $\R$. As it is well known, all atomless probability measures on $\R$ are equivalent, meaning that there is going to be an isomorphism of measurable spaces $f:\R\to\R$ for which $f_*p=q$. 
Now, let $D$ be any divergence on $\cat{Stoch}$. Since it is by construction an invariant of measurable spaces, also the entropy $H_D$ is, and so, necessarily,
\[
H_D(p) = H_D(q) .
\]
Another way to see this is that we can partition $\R$ into countably many measurable sets, and by the data processing inequality \eqref{dataprocH}, necessarily the entropies of both $p$ and $q$ have to be larger or equal than the entropies of the induced discrete distributions. But by partitioning $\R$ with an atomless measure, regardless of its variance, one can form all possible discrete probability distributions. 
By a similar argument, the divergence between Dirac deltas at distinct points $x,y\in\R$ does not depend on the distance between $x$ and $y$.

In general, the question to ask is,
\begin{question}
	Which category can one use instead of $\cat{Stoch}$ to accurately quantify the randomness of atomless distributions?
\end{question}

One possible answer to the question could be to choose a metric or divergence between distributions on a metric space which is sensitive to the underlying geometry. This would be a categorical counterpart to the growing interest, in the information theory community, for information \emph{geometry}~\cite{amari-foundation,infgeo}.
We leave the study of these metric-based Markov categories, and other more general finer invariants, to future work.

\appendix

\section{The category of divergence spaces}\label{Div}

Here we briefly mention the enriching category which gives divergence-enriched categories, by defining its morphisms and its closed monoidal structure. 

Recall divergence spaces from \Cref{defdiv}.
One can say that a divergence is to a metric as a reflexive relation (or graph) is to a preorder (or a reflexive, transitive graph).

\begin{definition}
 A \emph{morphism} of divergence spaces $f:X\to Y$ is a function which does not increase the divergences:
 $$
 D\big( f(x)\et f(x') \big) \le D(x\et x') .
 $$
\end{definition}

Denote by $\cat{Div}$ the category of divergence spaces and their morphisms.
The isomorphisms of $\cat{Div}$ are divergence-preserving bijections.

\begin{definition}
 Let $X$ and $Y$ be divergence spaces. We denote by $X\boxtimes Y$ the cartesian product $X\times Y$, together with the following divergence.
 $$
 D\big( (x,y)\et (x',y') \big) \coloneqq D(x\et x') + D(y\et y') .
 $$
\end{definition}

The monoidal unit is given by the one-point divergence space.

\begin{definition}
 Let $X$ and $Y$ be divergence spaces. We denote by $[X,Y]$ the set of morphisms $X\to Y$, together with the following divergence.
 $$
 D(f\et g) \coloneqq \max \left\{0,  \sup_{x,x'\in X} \Big( D\big(f(x)\et g(x')\big) - D(x\et x') \Big) \right\} . 
 $$
\end{definition}
This can be considered an enriched version of the ``graph exponential'' construction.

\begin{proposition}
 The category $\cat{Div}$, with the tensor product and internal hom defined above, is monoidal closed. 
\end{proposition}

The proof is similar to the case of graphs.

\begin{proof}
 Let $f:X\times Y\to Z$ be a function. Consider the ``curried'' function $f^\sharp:X\to Z^Y$ which maps $x\in X$ to 
 $$
 \begin{tikzcd}[row sep=0]
  Y \ar{r}{f^\sharp_x} & Z \\
  y \ar[mapsto]{r} &f^\sharp_x(y) = f(x,y) .
 \end{tikzcd}
 $$
 As $\cat{Set}$ is cartesian closed, we know that the assignment $f\mapsto f^\sharp$ gives a bijection $Z^{X\times Y}\cong (Z^Y)^X$. It remains to prove that $f$ is a morphism of divergence spaces if and only if $f^\sharp$ is, so that we get an isomorphism
 $$
 \cat{Div}(X\boxtimes Y, Z) \cong \cat{Div}(X, [Y,Z]) .
 $$
 So suppose $f$ is a morphism of divergence spaces. First of all, for all $x\in X$, we have that for all $y,y'\in Y$, 
 $$
 D \big( f^\sharp_x(y) \et f^\sharp_x(y') \big) = D\big( f(x,y)\et f(x,y')\big) \le D(y\et y') ,
 $$
 so that $f^\sharp_x\in [Y,Z]$. Moreover, for every $x,x'\in X$, we have that 
 \begin{align*}
 D(f^\sharp_x\et f^\sharp_{x'}) &\le \sup_{y,y'\in Y} \Big( D\big(f^\sharp_x(y)\et f^\sharp_{x'}(y')\big) - D(y\et y') \Big) \\
  &= \sup_{y,y'\in Y} \Big( D\big(f(x,y)\et f(x',y')\big) - D(y\et y') \Big) \\
  &\le \sup_{y,y'\in Y} \Big( D(x\et x') + D(y\et y') - D(y\et y') \Big) \\
  &= D(x\et x') ,
 \end{align*}
 so that $f^\sharp$ is a morphism of divergence spaces. 
 
 Conversely, suppose that $f^\sharp:X\to [Y, Z]$ is a morphism of divergence spaces. Then for every $x,x'\in X$ and $y,y'\in Y$,
 \begin{align*}
 D\big( f(x,y)\et f(x',y') \big) - D(y\et y') &= D \big( f^\sharp_x(y) \et f^\sharp_{x'}(y') \big) - D(y\et y') \\
  &\le \sup_{y,y'\in Y} \Big( D \big( f^\sharp_x(y) \et f^\sharp_{x'}(y') \big) - D(y\et y') \Big) \\
  &= D \big( f^\sharp_x \et f^\sharp_{x'} \big) \\
  &\le D(x\et x') ,
 \end{align*}
 so that 
 $$
 D\big( f(x,y)\et f(x',y') \big) \le D(x\et x') + D(y\et y') ,
 $$
 and so $f$ is a morphism of divergence spaces. 
\end{proof}

Therefore we can talk about $\cat{Div}$-categories, categories enriched in $\cat{Div}$.

Now, a category $\cat{C}$ is enriched in $\cat{Div}$ if all the hom-sets are equipped with a divergence, and moreover the composition maps
 \[
  \begin{tikzcd}[row sep=0]
  \cat{C}(A, B) \boxtimes \cat{C}(B, C) \ar{r}{\circ} & \cat{C}(A,C) \\
  (p, q) \ar[mapsto]{r} & q\circ p
  \end{tikzcd}
  \]
are divergence-nonincreasing.
We also say that a monoidal category $(C,\otimes, I)$ is \emph{monoidally} enriched in $\cat{Div}$ if moreover the tensor product maps
\[
   \begin{tikzcd}[row sep=0]
    \cat{C}(X,Y) \boxtimes \cat{C}(A,B) \ar{r}{\otimes} & \cat{C}(X\otimes Y, A \otimes B) \\
    (f,p) \ar[mapsto]{r} & f\otimes p
   \end{tikzcd}
  \]
are divergence-nonincreasing.
These are precisely the conditions appearing in \Cref{defdivmarkov}.

\bibliographystyle{alpha}
\bibliography{entropy}
\addcontentsline{toc}{section}{References}

\end{document}